\documentclass[letterpaper,12pt]{article}
\pdfoutput=1

\usepackage{setspace}
\usepackage{slashed}
\usepackage{epsfig}
\usepackage{comment}
\usepackage{cancel}
\usepackage{bbm}
\usepackage{array}
\usepackage{bigints}
\usepackage{booktabs}
\usepackage{xcolor}
\usepackage{dsfont}
\usepackage{float}
\usepackage{framed}
\usepackage{graphicx}
\usepackage{indentfirst}
\usepackage{mathrsfs}
\usepackage{multirow}
\usepackage{subdepth}
\usepackage{titlesec}
\usepackage[dotinlabels]{titletoc}
\usepackage{wrapfig}
\usepackage[vcentermath]{youngtab}
\usepackage{relsize}
\usepackage[colorlinks=true, urlcolor=dblue, linktocpage=true, linkcolor=dblue,citecolor=purple]{hyperref}
\usepackage{amsfonts}
\usepackage{amsthm}
\usepackage{amssymb}
\usepackage{cite}
\usepackage{mathtools}
\usepackage{cleveref}
\usepackage{epstopdf}
\usepackage{xspace}
\usepackage[normalem]{ulem}
\usepackage[utf8]{inputenc}

\numberwithin{equation}{section}

\crefname{section}{§}{§§}
\Crefname{section}{§}{§§}

\newtheorem{lemma}{Lemma}
\newtheorem{theorem}{Theorem}

\newtheorem{defn}{Definition}
\newtheorem{conj}{Conjecture}

\newtheorem{coro}{Corollary}

 \def\p{\partial}

\def\0{{(0)}}
\def\1{{(1)}}
\def\2{{(2)}}

\def\<{\langle }
\def\>{\rangle }

\newcommand{\bea}{\begin{eqnarray}}
\newcommand{\eea}{\end{eqnarray}}

   \makeatletter
  \let\over=\@@over \let\overwithdelims=\@@overwithdelims
  \let\atop=\@@atop \let\atopwithdelims=\@@atopwithdelims
  \let\above=\@@above \let\abovewithdelims=\@@abovewithdelims
\renewcommand\section{\@startsection {section}{1}{\z@}%
                                   {-3.5ex \@plus -1ex \@minus -.2ex}
                                   {2.3ex \@plus.2ex}%
                                   {\normalfont\large\bfseries}}

\renewcommand\subsection{\@startsection{subsection}{2}{\z@}%
                                     {-3.25ex\@plus -1ex \@minus -.2ex}%
                                     {1.5ex \@plus .2ex}%
                                     {\normalfont\bfseries}}

\renewcommand{\H}{\mathcal{H}}

\newcommand{\beq}{\begin{equation}}
\newcommand{\eeq}{\end{equation}}
\newcommand{\beqa}{\begin{eqnarray}}
\newcommand{\eeqa}{\end{eqnarray}}
\newcommand{\beqar}{\begin{eqnarray*}}

\def\[{\big[}
\def\]{\big]}
\def\la{\langle}
\def\ra{\rangle}

\def\p{\partial}

\colorlet{dblue}{blue!70!black}

\newcommand\xy{{\langle xy \rangle}}

\newcommand{\pd}{\partial}

\newcommand\be{\begin{equation}}
\newcommand\ee{\end{equation}}
\newcommand{\ba}{\begin{eqnarray}}
\newcommand{\ea}{\end{eqnarray}}

\newcommand{\lb}{\left(}
\newcommand{\rb}{\right)}
\newcommand{\lsb}{\left[}
\newcommand{\rsb}{\right]}

\newcommand{\ket}[1]{|{#1}\rangle}

\DeclareMathOperator{\area}{area}
\DeclareMathOperator{\vol}{vol}

\newcommand{\GN}{4G_{\rm N}}

\def\Res{\text{Res}}
\def\interior#1{%
  {\kern0pt#1}^{\mathrm{o}}%
}
\newcommand{\AB}{inner-su\-perbal\-anced\xspace}
\def\trick{Flow extension lemma}

\begin{document}
\begin{spacing}{1.3}
\begin{titlepage}

\renewcommand*{\thefootnote}{\fnsymbol{footnote}}

\ \\
\vspace{-2.3cm}
\begin{center}

{\LARGE{\textsc{Bit Threads and Holographic Monogamy}}}

\vspace{0.5cm}
Shawn X. Cui,$^1$\footnote{Present affiliation: Department of Mathematics, Virginia Tech, Blacksburg, VA 24061, USA} Patrick Hayden,$^1$ Temple He,$^2$\footnote{Present affiliation: Center for Quantum Mathematics and Physics (QMAP), Department of Physics, University of California, Davis, CA 95616, USA} Matthew Headrick,$^{3,4}$ Bogdan Stoica,$^{3,5}$ and Michael Walter$^6$

\vspace{5mm}

{\small
\textit{
$^1$Stanford Institute for Theoretical Physics, Stanford University, Stanford, CA 94305,~USA }\\

\vspace{2mm}

\textit{
$^2$Center for the Fundamental Laws of Nature, Harvard University, Cambridge, MA 02138, USA }\\

\vspace{2mm}

\textit{$^3$Martin A. Fisher School of Physics, Brandeis University, Waltham, MA 02453, USA}\\

\vspace{2mm}

\textit{$^4$Center for Theoretical Physics, Massachusetts Institute of Technology, Cambridge, MA 02139, USA}\\

\vspace{2mm}

\textit{
$^5$Department of Physics, Brown University, Providence, RI 02912, USA}\\

\vspace{2mm}

\textit{
$^6$Korteweg-de Vries Institute for Mathematics, Institute of Physics, Institute for Logic, Language \& Computation, and QuSoft, University of Amsterdam, 1098 XG Amsterdam, The Netherlands}\\

\vspace{4mm}

{\tt xingshan@vt.edu,   phayden@stanford.edu, tmhe@ucdavis.edu, headrick@brandeis.edu, bstoica@brandeis.edu, m.walter@uva.nl}

\vspace{0.3cm}
}

\end{center}

\begin{abstract}
Bit threads provide an alternative description of holographic entanglement, replacing the Ryu-Takayanagi minimal surface with bulk curves connecting pairs of boundary points.
We use bit threads to prove the monogamy of mutual information (MMI) property of holographic entanglement entropies.
This is accomplished using the concept of a so-called multicommodity flow, adapted from the network setting, and tools from the theory of convex optimization.
Based on the bit thread picture, we conjecture a general ansatz for a holographic state, involving only bipartite and perfect-tensor type entanglement, for any decomposition of the boundary into four regions.
We also give new proofs of analogous theorems on networks.
\end{abstract}

\vfill

\begin{flushleft}{\small BRX-6330, Brown HET-1764, MIT-CTP/5036}\end{flushleft}

\end{titlepage}

\setcounter{tocdepth}{2}
\tableofcontents

\renewcommand*{\thefootnote}{\arabic{footnote}}
\setcounter{footnote}{0}

\section{Introduction}\label{sec:intro}

One of the most important relationships between holographic gravity and entanglement is the Ryu-Takayanagi (RT) formula, which states that entanglement entropy of a region in the boundary conformal field theory (CFT) is dual to a geometric extremization problem in the bulk~\cite{Ryu:2006bv,Ryu:2006ef}.
Specifically, the formula states that the entropy of a spatial region $A$ on the boundary CFT is given by
\begin{align}
	S(A) = \frac{1}{\GN}\area(m(A))\ ,
\end{align}
where $m(A)$ is a minimal hypersurface in the bulk homologous to $A$. This elegant formula is essentially an anti-de Sitter (AdS) cousin of the black hole entropy formula, but more importantly, it is expected to yield new insights toward how entanglement and quantum gravity are connected~\cite{VanRaamsdonk:2010pw,Maldacena:2013xja}.

Despite the fact that the RT formula has been a subject of intense research for over a decade, there are still many facets of it that are only now being discovered. Indeed, only recently was it demonstrated that the geometric extremization problem underlying the RT formula can alternatively be interpreted as a flow extremization problem~\cite{Freedman:2016zud,Headrick:2017ucz}. By utilizing the Riemannian version of the max flow-min cut theorem, it was shown that the maximum flux out of a boundary region $A$, optimized over all divergenceless bounded vector fields in the bulk, is precisely the area of~$m(A)$. Because this interpretation of the RT formula suggests that the vector field captures the maximum information flow out of region $A$, the flow lines in the vector field became known as ``bit threads.'' These bit threads are a tangible geometric manifestation of the entanglement between $A$ and its complement.

Although bit threads paint an attractive picture that appears to capture more intuitively the information-theoretic meaning behind holographic entanglement entropy, there is still much not understood about them. They were used to provide alternative proofs of subadditivity and strong subadditivity in~\cite{Freedman:2016zud}, but the proof of the monogamy of mutual information (MMI) remained elusive. MMI is an inequality which, unlike subadditivity and strong subadditivity, does not hold for general quantum states, but is obeyed for holographic systems in the semiclassical or large $N$ limit. It is given by
\begin{equation}
\label{I3isneg}
\begin{aligned}
-I_3(A:B:C) :=\ &S(AB) + S(AC) + S(BC) \\
&-S(A) -S(B) - S(C) - S(ABC)   \ge 0\ .
\end{aligned}
\end{equation}
The quantity~$-I_3$ is known as the (negative) tripartite information, and property~\eqref{I3isneg} was proven in~\cite{Hayden:2011ag,Headrick:2013zda} using minimal surfaces.\footnote{MMI was also proven in the covariant setting in~\cite{Wall:2012uf}.} While MMI is a general fact about holographic states, the reason for this from a more fundamental viewpoint is not clear.
Presumably, such states take a special form which guarantees MMI (cf.~\cite{Nezami:2016zni,ding2016conditional}).
What is this form?
It was suggested in~\cite{Freedman:2016zud} that understanding MMI from the viewpoint of bit threads may shed some light on this question.

In this paper we will take up these challenges. First, we will provide a proof of MMI based on bit threads. Specifically, we show that, given a decomposition of the boundary into regions, there exists a thread configuration that simultaneously maximizes the number of threads connecting each region to its complement. MMI follows essentially directly from this statement.\footnote{V.\ Hubeny has given a method to explicitly construct such a thread configuration, thereby establishing MMI, in certain cases~\cite{Hubeny:2018bri}.} This theorem is the continuum analogue of a well-known result in the theory of multicommodity flows on networks. However, the standard network proof is discrete and combinatorial in nature and is not straightforwardly adapted to the continuum. Therefore, we develop a new method of proof based on strong duality of convex programs. Convex optimization proofs have the advantage that they work in essentially the same way on graphs and Riemannian manifolds, whereas the graph proofs standard in the literature often rely on integer edge capacities, combinatorics, and other discrete features, and do not readily translate over to the continuous case.\footnote{Conversely, when additional structure is present, such as integer capacity edges in a graph, the statements that can be proven are often slightly stronger than what can be proven in the absence of such extra structure, e.g.\ by also obtaining results on the integrality of the flows.} The convex optimization methods offer a unified point of view for both the graph and Riemannian geometry settings, and are a stand-alone mathematical result. As far as we know, these are the first results on multicommodity flows on Riemannian manifolds.

Second, we use the thread-based proof of MMI to motivate a particular entanglement structure for holographic states, which involves pairwise-entangled states together with a four-party state with perfect-tensor entanglement (cf.~\cite{Nezami:2016zni}). MMI is manifest in this ansatz, so if it is correct then it explains why holographic states obey MMI\@.

It has also been proven that holographic entropic inequalities exist for more than four boundary regions~\cite{Bao:2015bfa}. For example, MMI is part of a family of holographic entropic inequalities with dihedral symmetry in the boundary regions. These dihedral inequalities exist for any odd number of boundary regions, and for five regions other holographic inequalities are also known. However, the general structure of holographic inequalities for more than five boundary regions is currently not known. It would be interesting to try to understand these inequalities from the viewpoint of bit threads. In this paper, we make a tentative suggestion for the general structure of holographic states in terms of the extremal rays of the so-called holographic entropy cone.

We organize the paper in the following manner. In Section~\ref{sec:bbthrev}, we give the necessary background on holographic entanglement entropy, flows, bit threads, MMI, and related notions. In Section~\ref{sec:continuum}, we state the main theorem in this paper concerning the existence of a maximizing thread configuration on multiple regions and show that MMI follows from it. In Section~\ref{sec:statedecomp}, we use bit threads and the proof of MMI to motivate the conjecture mentioned above concerning the structure of holographic states. In Section~\ref{sec:proofs}, we prove our main theorem as well as a useful generalization of it. Section~\ref{sec:discrete} revisits our continuum results in the graph theoretic setting, demonstrating how analogous arguments can be developed there. In Section~\ref{sec:conclusion} we discuss open issues.

\section{Background}\label{sec:bbthrev}

\subsection{Ryu-Takayanagi formula and bit threads}\label{sec:RT}

We begin with some basic concepts and definitions concerning holographic entanglement entropies. In this paper, we work in the regime of validity of the Ryu-Takayanagi formula, namely a conformal field theory dual to Einstein gravity in a state represented by a classical spacetime with a time-reflection symmetry. The Cauchy slice invariant under the time reflection is a Riemannian manifold that we will call $\mathcal M$. We assume that a cutoff has been introduced ``near'' the conformal boundary so that $\mathcal M$ is a compact manifold with boundary. Its boundary $\partial \mathcal M$ is the space where the field theory lives.

It is sometimes convenient to let the bulk be bounded also on black hole horizons, thereby representing a thermal mixed state of the field theory. However, for definiteness in this paper we will consider only pure states of the field theory, and correspondingly for us $\partial \mathcal M$ will \emph{not} include any horizons.\footnote{$\mathcal M$ may have an ``internal'' boundary $\mathcal{B}$ that does not carry entropy, such as an orbifold fixed plane or end-of-the-world brane. This is accounted for in the Ryu-Takayanagi formula~\eqref{RT} by defining the homology to be relative to $\mathcal{B}$, and in the max flow formula~\eqref{maxflow} by requiring the flow $v^\mu$ to satisfy a Neumann boundary condition $n_\mu v^\mu=0$ along $\mathcal{B}$, and in the bit thread formula~\eqref{maxflow} by not allowing threads to end on $\mathcal{B}$. See~\cite{Headrick:2017ucz} for a fuller discussion. While we will not explicitly refer to internal boundaries in the rest of this paper, all of our results are valid in the presence of such a boundary.\label{internalboundary}} This assumption is without loss of generality, since it is always possible to purify a thermal state by passing to the thermofield double, which is represented holographically by a two-sided black hole.

Let $A$ be a region of $\partial \mathcal M$. The Ryu-Takayanagi formula~\cite{Ryu:2006bv,Ryu:2006ef} then gives its entropy $A$ as the area of the minimal surface in $\mathcal M$ homologous to $A$ (relative to $\partial A$):
\begin{equation}\label{RT}
S(A) = \frac1\GN\min_{m\sim A}\area(m)\ .
\end{equation}
(We could choose to work in units where $\GN=1$, and this would simplify certain formulas, but it will be useful to maintain a clear distinction between the microscopic Planck scale $G_{\rm N}^{1/(d-1)}$ and the macroscopic scale of $\mathcal M$, defined for example by its curvatures.) We will denote the minimal surface by $m(A)$\footnote{The minimal surface is generically unique. In cases where it is not, we let $m(A)$ denote any choice of minimal surface.} and the corresponding homology region, whose boundary is $A\cup m(A)$, by $r(A)$.
The homology region is sometimes called the ``entanglement wedge'', although strictly speaking the entanglement wedge is the causal domain of the homology region.

\subsubsection{Flows}

The notion of bit threads was first explored in~\cite{Freedman:2016zud}. To explain them, we first define a \emph{flow}, which is a vector field $v$ on $\mathcal M$ that is divergenceless and has norm bounded everywhere by $1/\GN$:\footnote{Flows are equivalent, via the Hodge star, to $(d-1)$-calibrations \cite{HL}: the $(d-1)$-form $\omega=*(4G_{\rm N}g_{\mu\nu}v^\mu dx^\nu)$ is a calibration if and only if the vector field $v$ is a flow. See \cite{Bakhmatov:2017ihw} for further work using calibrations to compute holographic entanglement entropies.}
\begin{align}
\label{flowdef}
	\nabla\cdot v = 0\ , \quad |v| \leq \frac1\GN\ .
\end{align}
For simplicity we denote the flux of a flow $v$ through a boundary region $A$ by $\int_A v$:
\begin{equation}\label{flux}
\int_A v:=\int_A\sqrt h\,\hat n\cdot v\ ,
\end{equation}
where $h$ is the determinant of the induced metric on $A$ and $\hat n$ is the (inward-pointing) unit normal vector. The flow $v$ is called a max flow on $A$ if the flux of $v$ through $A$ is maximal among all flows. We can then write the entropy of $A$ as the flux through $A$ of a max flow:
\begin{equation}\label{maxflow}
S(A) = \max_{v\text{ flow}}\int_A v\ .
\end{equation}

The equivalence of~\eqref{maxflow} to the RT formula~\eqref{RT} is guaranteed by the Riemannian version of the max flow-min cut theorem~\cite{Federer74,MR700642,MR1088184,MR2685608,Headrick:2017ucz}:
\begin{equation}\label{mfmc}
\max_{v\text{ flow}}\int_A v = \frac1\GN\min_{m\sim A}\area(m)\ .
\end{equation}
The theorem can be understood heuristically as follows: by its divergencelessness, $v$ has the same flux through every surface homologous to $A$, and by the norm bound this flux is bounded above by its area. The strongest bound is given by the minimal surface, which thus acts as the bottleneck limiting the flow. The fact that this bound is tight is proven by writing the left- and right-hand sides of~\eqref{mfmc} in terms of convex programs and invoking strong duality to equate their solutions. (See~\cite{Headrick:2017ucz} for an exposition of the proof.) While the minimal surface $m(A)$ is typically unique, the maximizing flow is typically highly non-unique; on the minimal surface it equals $1/\GN$ times the unit normal vector, but away from the minimal surface it is underdetermined.

\subsubsection{Bit threads}

We can further rewrite~\eqref{maxflow} by thinking about the integral curves of a flow $v$, in the same way that it is often useful to think about electric and magnetic field lines rather than the vector fields themselves.
We can choose a set of integral curves whose transverse density equals $|v|$ everywhere.
In~\cite{Freedman:2016zud} these curves were called \emph{bit threads}.

The integral curves of a given vector field are oriented and locally parallel. It will be useful to generalize the notion of bit threads by dropping these two conditions.
Thus, in this paper, the threads will be unoriented curves, and we will allow them to pass through a given neighborhood at different angles and even to intersect.
Since the threads are not locally parallel, we replace the notion of \emph{transverse density} with simply \emph{density}, defined at a given point as the total length of the threads in a ball of radius $R$ centered on that point divided by the volume of the ball, where $R$ is chosen to be much larger than the Planck scale $G_{\rm N}^{1/(d-1)}$ and much smaller than the curvature scale of $\mathcal M$.%
\footnote{It is conceptually natural to think of the threads as being microscopic but discrete, so that for example we can speak of the \emph{number} of threads connecting two boundary regions.
To be mathematically precise one could instead define a thread configuration as a continuous set $\{\mathcal{C}\}$ of curves equipped with a measure $\mu$.
The density bound would then be imposed by requiring that, for every open subset $s$ of $\mathcal{M}$, $\int d\mu \,\text{length}(\mathcal{C}\cap s)\le\vol(s)/\GN$, and the ``number'' of threads connecting two boundary regions would be defined as the total measure of that set of curves.}
A \emph{thread configuration} is thus defined as a set of unoriented curves on $\mathcal{M}$ obeying the following rules:
\begin{enumerate}\label{thread rules}
\item Threads end only on $\partial \mathcal M$.
\item The thread density is nowhere larger than $1/\GN$.
\end{enumerate}
A thread can be thought of as the continuum analogue of a ``path'' in a network, and a thread configuration is the analogue of a set of edge-disjoint paths, a central concept in the analysis of network flows.

Given a flow $v$, we can, as noted above, choose a set of integral curves with density $|v|$; dropping their orientations yields a thread configuration. In the classical or large-$N$ limit $G_{\rm N}\to0$, the density of threads is large on the scale of $\mathcal M$ and we can neglect any discretization error arising from replacing the continuous flow $v$ by a discrete set of threads. Thus a flow maps essentially uniquely (up to the unimportant Planck-scale choice of integral curves) to a thread configuration. However, this map is not invertible: a given thread configuration may not come from any flow, since the threads may not be locally parallel, and even if such a flow exists it is not unique since one must make a choice of orientation for each thread. The extra flexibility afforded by the threads is useful since, as we will see in the next section, a single thread configuration can simultaneously represent several different flows. On the other hand, the flows are easier to work with technically, and in particular we will use them as an intermediate device for proving theorems about threads; an example is~\eqref{maxthread} below.

We denote the number of threads connecting a region $A$ to its complement $\bar A:=\partial \mathcal M\setminus A$ in a given configuration by $N_{A\bar A}$. We will now show that the maximum value of $N_{A\bar A}$ over allowed configurations is $S(A)$:
\begin{equation}\label{maxthread}
S(A) = \max N_{A\bar A}\ .
\end{equation}
First, we will show that $N_{A\bar A}$ is bounded above by the area of any surface $m\sim A$ divided by $\GN$. Consider a slab of thickness $R$ around $m$ (where again $R$ is much larger than the Planck length and much smaller than the curvature radius of $\mathcal{M}$); this has volume $R\area(m)$, so the total length of all the threads within the slab is bounded above by $R\area(m)/\GN$. On the other hand, any thread connecting $A$ to $\bar A$ must pass through $m$, and therefore must have length within the slab at least $R$. So the total length within the slab of all threads connecting $A$ to $\bar A$ is at least $RN_{A\bar A}$. Combining these two bounds gives
\begin{equation}
N_{A\bar A}\le\frac1\GN\area(m)\ .
\end{equation}
In particular, for the minimal surface $m(A)$,
\begin{equation}\label{Nbound}
N_{A\bar A}\le\frac1\GN\area(m(A))=S(A)\ .
\end{equation}
Again,~\eqref{Nbound} applies to \emph{any} thread configuration. On the other hand, as described above, given any flow $v$ we can construct a thread configuration by choosing a set of integral curves whose density equals $|v|$ everywhere. The number of threads connecting $A$ to $\bar A$ is at least as large as the flux of $v$ on $A$:
\begin{equation}\label{Nfluxbound}
N_{A\bar A}\ge\int_A v\ .
\end{equation}
The reason we don't necessarily have equality is that some of the integral curves may go from $\bar A$ to $A$, thereby contributing negatively to the flux but positively to $N_{A\bar A}$. Given~\eqref{Nbound}, however, for a max flow $v(A)$ this bound must be saturated:
\begin{equation}\label{Nboundtight}
N_{A\bar A}=\int_A v(A)=S(A)\ .
\end{equation}

The bit threads connecting $A$ to $\bar A$ are vivid manifestations of the entanglement between $A$ and $\bar A$, as quantified by the entropy $S(A)$. This viewpoint gives an alternate interpretation to the RT formula that may in many situations be more intuitive.
For example, given a spatial region $A$ on the boundary CFT, the minimal hypersurface homologous to $A$ does not necessarily vary continuously as $A$ varies: an infinitesimal perturbation of $A$ can result in the minimal hypersurface changing drastically, depending on the geometry of the bulk. Bit threads, on the other hand, vary continuously as a function of $A$, even when the bottleneck surface jumps.

Heuristically, it is useful to visualize each bit thread as defining a ``channel'' that allows for one bit of (quantum) information to be communicated between different regions on the spatial boundary. The amount of information that can be communicated between two spatially separated boundary regions is then determined by the number of channels that the bulk geometry allows between the two regions. Importantly, whereas the maximizing bit thread configuration may change depending on the boundary region we choose, the set of all allowable configurations is completely determined by the geometry.
The ``channel'' should be viewed as a metaphor, however, similar to how a Bell pair can be be viewed as enabling a channel in the context of teleportation.
While it is known that Bell pairs can always be distilled at an optimal rate~$S(A)$, we conjecture a more direct connection between bit threads and the entanglement structure of the the underlying holographic states, elaborated in Section~\ref{sec:statedecomp}.

\subsubsection{Properties and derived quantities}\label{sec:properties}

Many interesting properties of entropies and quantities derived from them can be written naturally in terms of flows or threads. For example, let $A$, $B$ be disjoint boundary regions, and let $v$ be a max flow for their union, so $\int_{AB}v=S(AB)$. Then we have, by~\eqref{maxflow}
\begin{equation}
S(A) \ge \int_A v\ ,\qquad S(B) \ge\int_B v\ ,
\end{equation}
hence
\begin{equation}
S(A)+S(B)\ge\int_A v+\int_B v=\int_{AB}v=S(AB)\ ,
\end{equation}
which is the subadditivity property.

A useful property of flows is that there always exists a flow that simultaneously maximizes the flux through $A$ and $AB$ (or $B$ and $AB$, but not in general $A$ and $B$). We call this the \emph{nesting property}, and it is proven in~\cite{Headrick:2017ucz}. Let $v_1$ be such a flow. (In the notation of \cite{Freedman:2016zud}, $v_1$ would be called $v(A,B)$.) We then obtain the following formula for the conditional information:
\begin{equation}
H(B|A):=S(AB)-S(A)=\int_{AB}v_1-\int_A v_1=\int_B v_1\ .
\end{equation}
We can also write this quantity in terms of threads. Let $C$ be the complement of $AB$, and let $N^1_{AB}$, $N^1_{AC}$, $N^1_{BC}$ be the number of threads connecting the different pairs of regions in the flow $v_1$.\footnote{In addition to the threads connecting distinct boundary regions, there may be threads connecting a region to itself or simply forming a loop in the bulk. These will not play a role in our considerations.} Using~\eqref{Nboundtight}, we then have
\begin{equation}
S(AB) = N^1_{AC}+N^1_{BC}\ ,\qquad
S(A) = N^1_{AC}+N^1_{AB}\ .
\end{equation}
(Note that we don't have a formula for $S(B)$ in terms of these threads, since the configuration does not maximize the number connecting $B$ to its complement.) Hence
\begin{equation}\label{threadCE}
H(B|A) = N^1_{BC} - N^1_{AB}\ .
\end{equation}
For the mutual information, we have
\begin{equation}
I(A:B):=S(A)+S(B)-S(AB) = \int_A\left(v_1-v_2\right) = N^2_{AB}+N^2_{BC}+N^1_{AB}-N^1_{BC}\ ,
\end{equation}
where $v_2$ is a flow with maximum flux through $B$ and $AB$, and $N^2_{ij}$ are the corresponding numbers of threads. In the next section, using the concept of a \emph{multiflow}, we will write down a more concise formula for the mutual information in terms of threads (see~\eqref{threadMI}).

The nesting property also allows us to prove the strong subadditivity property, $S(AB)+S(BC)\ge S(B)+S(ABC)$, where $A$, $B$, $C$ are disjoint regions. (Unlike in the previous paragraph, here $C$ is not necessarily the complement of $AB$, i.e.\ $ABC$ does not necessarily cover $\partial\mathcal M$.) Let $v$ be a flow that maximizes the flux through both $B$ and $ABC$. Then
\begin{equation}
S(AB)\ge \int_{AB}v\ ,\qquad
S(BC)\ge \int_{BC}v\ ,
\end{equation}
hence
\begin{equation}
S(AB)+S(BC) \ge \int_{AB}v+\int_{BC}v = \int_B v+\int_{ABC}v = S(B)+S(ABC)\ .
\end{equation}

\subsection{MMI, perfect tensors, and entropy cones}\label{sec:MMI}

\subsubsection{MMI}

Given three subsystems $A,B,C$ of a quantum system, the \emph{(negative) tripartite information} is defined as the following linear combination of the subsystem entropies:\footnote{One can alternatively work with the quantity $I_3$, defined as the negative of $-I_3$. However, when discussing holographic entanglement entropies, $-I_3$ is more convenient since it is non-negative.}
\begin{align}
-I_3(A:B:C) &:= S(AB) + S(BC) + S(AC) - S(A) - S(B) - S(C) - S(ABC) \nonumber \\
&= I(A:BC) - I(A:B) - I(A:C)\ .
\end{align}
The quantity~$-I_3$ is manifestly symmetric under permuting $A,B,C$. In fact it is even more symmetric than that; defining $D:=\overline{ABC}$, it is symmetric under the full permutation group on $A,B,C,D$. (Note that, by purity, $S(AD)=S(BC)$, $S(BD)=S(AC)$, and $S(CD)=S(AB)$.) Since in this paper we will mostly be working with a fixed set of 4 parties, we will usually simply write $-I_3$, without arguments.

We note that~$-I_3$ is sensitive only to fully four-party entanglement, in the following sense. If the state is unentangled between any party and the others, or between any two parties and the others, then $-I_3$ vanishes. In a general quantum system, it can be either positive or negative. For example, in the four-party GHZ state,
\begin{equation}\label{GHZ}
\ket{\psi}_{ABCD} = \frac1{\sqrt2}\left(\ket{0000}+\ket{1111}\right)\ ,
\end{equation}
it is negative: $-I_3=-\ln2$. On the other hand, for the following state (with $D$ being a 4-state system),
\begin{equation}
\ket{\psi}_{ABCD} = \frac12\left(\ket{0000}+\ket{0111}+\ket{1012}+\ket{1103}\right)
\end{equation}
it is positive: $-I_3=\ln2$.

$-I_3$ is also positive for a class of four-party states called perfect-tensor states, which will play an important role in our considerations. A perfect-tensor state is a pure state on $2n$ parties such that the reduced density matrix on any $n$ parties is maximally mixed. For four parties, this implies that all the one-party entropies are equal, and all the two-party entropies have twice that value~\cite{Nezami:2016zni}:
\begin{equation}\label{PT_cond}
S(A)=S(B)=S(C)=S(D)=S_0\ ,\quad
S(AB)=S(BC)=S(AC)=2S_0\ ,
\end{equation}
where $S_0>0$.
Hence
\begin{equation}
-I_3 = 2S_0>0\ .
\end{equation}
In this paper, we will use the term \emph{perfect tensor (PT)} somewhat loosely to denote a four-party pure state whose entropies take the form~\eqref{PT_cond} for some $S_0>0$, even if they are not maximal for the respective Hilbert spaces.

In a general field theory, with the subsystems $A,B,C$ being spatial regions, $-I_3$ can take either sign \cite{Casini:2008wt}. However, it was proven in~\cite{Hayden:2011ag,Headrick:2013zda} that the entropies derived from the RT formula always obey the inequality
\begin{equation}\label{MMI}
-I_3(A:B:C) \ge 0\ ,
\end{equation}
which is known as \emph{monogamy of mutual information} (MMI). The proof involved cutting and pasting minimal surfaces. In this paper we will provide a proof of MMI based on flows or bit threads. Since a general state of a four-party system does not obey MMI, classical states of holographic systems (i.e.\ those represented by classical spacetimes) must have a particular entanglement structure in order to always obey MMI\@. It is not known what that entanglement structure is, and another purpose of this paper is to address this question.

\subsubsection{Entropy cones}\label{sec:entropy cones}

A general four-party pure state has 7 independent entropies, namely the 4 one-party entropies~$S(A)$, $S(B)$, $S(C)$, $S(D)$, together with 3 independent two-party entropies, e.g.\ $S(AB)$, $S(AC)$, and $S(BC)$. This set of numbers defines an \emph{entropy vector} in~$\mathbb R^7$. There is a additive structure here because entropies add under the operation of combining states by the tensor product. In other words, if
\begin{equation}
\ket{\psi}_{ABCD} = \ket{\psi_1}_{A_1B_1C_1D_1}\otimes\ket{\psi_2}_{A_2B_2C_2D_2}\ ,
\end{equation}
with $\H_A=\H_{A_1}\otimes\H_{A_2}$ etc., then the entropy vector of $\ket{\psi}$ is the sum of those of $\ket{\psi_1}$ and $\ket{\psi_2}$. The inequalities that the entropies satisfy---non-negativity, subadditivity, and strong subadditivity---carve out a set of possible entropy vectors which (after taking the closure) is a convex polyhedral cone in $\mathbb R^7$, called the four-party \emph{quantum entropy cone}. Holographic states satisfy MMI in addition to those inequalities, carving out a smaller cone, called the four-party \emph{holographic entropy cone}~\cite{Bao:2015bfa}. It is a simple exercise in linear algebra to show that the six pairwise mutual informations together with $-I_3$ also form a coordinate system (or dual basis) for $\mathbb R^7$:
\begin{equation}\label{vecbasis}
I(A:B)\ ,\; \;
I(A:C)\ ,\; \;
I(A:D)\ ,\; \;
I(B:C)\ ,\; \;
I(B:D)\ ,\; \;
I(C:D)\ ,\; \;
-I_3\ .
\end{equation}
For any point in the holographic entropy cone, these 7 quantities are non-negative---the mutual informations by subadditivity, and $-I_3$ by MMI\@. In fact, the converse also holds. Since MMI and subadditivity imply strong subadditivity, any point in $\mathbb R^7$ representing a set of putative entropies such that all 7 linear combinations~\eqref{vecbasis} are non-negative also obeys all the other inequalities required of an entropy, and is therefore in the holographic entropy cone. In other words, using the 7 quantities~\eqref{vecbasis} as coordinates, the holographic entropy cone consists precisely of the non-negative orthant in $\mathbb R^7$.

Any entropy vector such that exactly one of the 7 coordinates~\eqref{vecbasis} is positive, with the rest vanishing, is an \emph{extremal vector} of the holographic entropy cone; it lies on a 1-dimensional edge of that cone. Since the cone is 7-dimensional, any point in the cone can be written uniquely as a sum of 7 (or fewer) extremal vectors, one for each edge. States whose entropy vectors are extremal are readily constructed: a state with $I(A:B)>0$ and all other quantities in~\eqref{vecbasis} vanishing is necessarily of the form $\ket{\psi_1}_{AB}\otimes\ket{\psi_2}_C\otimes\ket{\psi_3}_D$, and similarly for the other pairs; while a state with $-I_3>0$ and all pairwise mutual informations vanishing is necessarily a PT\@. It is also possible to realize such states, and indeed arbitrary points in the holographic entropy cone, by holographic states~\cite{Balasubramanian:2014hda,Bao:2015bfa}.

\begin{figure}
\centering
\includegraphics[scale=0.7]{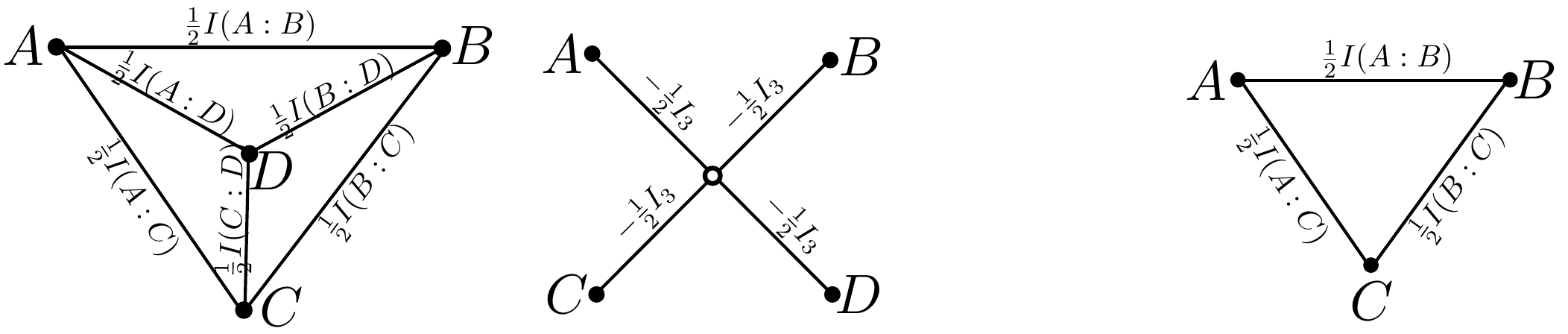}
\caption{Left: Skeleton graph (with two connected components) representing an arbitrary entropy vector in the four-party holographic entropy cone.
Right: Skeleton graph for a three-party entropy vector.}\label{fig:skeleton}
\end{figure}

The extremal rays can also be represented as simple graphs with external vertices $A,B,C,D$. For example, a graph with just a single edge connecting $A$ and $B$ with capacity $c$ gives an entropy vector with $I(A:B)=2c$ and all other coordinates in~\eqref{vecbasis} vanishing. Similarly, a star graph with one internal vertex connected to all four external vertices and capacity $c$ on each edge gives an entropy vector with $-I_3=2c$ and all other coordinates in~\eqref{vecbasis} vanishing. Since any vector in the holographic cone can be uniquely decomposed into extremal rays, it is reproduced by a (unique) ``skeleton graph'' consisting of the complete graph on $\{A,B,C,D\}$ with capacity $\frac12I(A:B)$ on the edge connecting $A$ and $B$ and similarly for the other pairs, plus a star graph with capacity $-\frac12I_3$ on each edge. This is shown in Fig.~\ref{fig:skeleton}.

Let us briefly discuss the analogous situation for states on fewer or more than four parties. For a three-party pure state, there are only 3 independent entropies (since $S(AB)=S(C)$ etc.), so the entropy vector lives in $\mathbb R^3$. Holographic states obey no extra entropy inequalities beyond those obeyed by any quantum state, namely non-negativity and subadditivity, so the holographic entropy cone is the same as the quantum entropy cone. A dual basis is provided by the three mutual informations,
\begin{equation}
I(A:B)\ ,\quad I(B:C)\ ,\quad I(A:C)\ .
\end{equation}
There are 3 types of extremal vectors, given by two-party entangled pure states $\ket{\psi_1}_{AB}\otimes\ket{\psi_2}_C$ etc. Thus the skeleton graph for three parties is simply a triangle, shown in Fig.~\ref{fig:skeleton}.

Given a decomposition of the boundary into four regions, we can merge two of the regions, say $C$ and $D$, and thereby consider the same state as a three-party pure state. Under this merging, the four-party skeleton graph on the left side of Fig.~\ref{fig:skeleton} turns into the three-party one on the right side as follows. The star graph splits at the internal vertex to become two edges, an $A(CD)$ edge and a $B(CD)$ edge, each with capacity $-\frac12I_3$. The first merges with the $AC$ and $AD$ edges from the four-party complete graph to give an $A(CD)$ edge with total capacity
\begin{equation}
-\frac12I_3+\frac12I(A:C)+\frac12I(A:D) = \frac12I(A:CD)\ .
\end{equation}
Similarly for the $B(CD)$ edges. The $AB$ edge remains unchanged and the $CD$ edge is removed. This rearrangement will play a role in our considerations of Section~\ref{sec:statedecomp}.

The situation for five and more parties was studied in~\cite{Bao:2015bfa,Cuenca:2019uzx}. For five-party pure states, there are no new inequalities beyond MMI and the standard quantum ones (non-negativity, subadditivity, strong subadditivity). There are 20 extremal vectors of the holographic entropy cone, given by 10 two-party entangled pure states, 5 four-party PTs, and 5 six-party PTs with two of the parties merged (e.g.\ a PT on $A_1A_2BCDE$ with $A=A_1A_2$). Since the cone is only 15-dimensional, the decomposition of a generic point into a sum of extremal vectors is \emph{not} unique, unlike for three or four parties. For six-party pure states, there are new inequalities; a complete list of inequalities was conjectured in~\cite{Bao:2015bfa} and proved in \cite{Cuenca:2019uzx}. Notable is the fact that the extremal rays are no longer only made from perfect tensors; rather, new entanglement structures come into play. For more than six parties, some new inequalities are known but a complete list has not even been conjectured.

\section{Multiflows and MMI}\label{sec:continuum}

As we reviewed in Subsection~\ref{sec:properties}, the subadditivity and strong subadditivity inequalities can be proved easily from the formula~\eqref{maxflow} for the entropy in terms of flows. Subadditivity follows more or less directly from the definition of a flow, while strong subadditivity requires the nesting property for flows (existence of a simultaneous max flow for $A$ and $AB$). Holographic entanglement entropies also obey the MMI inequality~\eqref{MMI}, which was proven using minimal surfaces~\cite{Hayden:2011ag,Headrick:2013zda}. Therefore it seems reasonable to expect MMI to admit a proof in terms of flows. However, it was shown in~\cite{Freedman:2016zud} that the nesting property alone is not powerful enough to prove MMI\@. Therefore, flows must obey some property beyond nesting. In this section we will state the necessary property and give a flow-based proof of MMI\@. The property is the existence of an object called a \emph{max multiflow}. It is guaranteed by our Theorem~\ref{thm17}, stated below and proved in Section~\ref{sec:proofs}.

\subsection{Multiflows}\label{sec:multiflows}

It turns out that the property required to prove MMI concerns not a single flow, like the nesting property, but rather a collection of flows that are compatible with each other in the sense that they can simultaneously occupy the same geometry (we will make this precise below). In the network context, such a collection of flows is called a \emph{multicommodity flow}, or \emph{multiflow}, and there is a large literature about them. (See Section~\ref{sec:discrete} for the network definition of a multiflow. Standard references are~\cite{frank1997integer,schrijver2003combinatorial}; two resources we have found useful are~\cite{Chandra,Guyslain}.) We will adopt the same terminology for the Riemannian setting we are working in here. We thus start by defining a multiflow.

\begin{defn}[Multiflow]\label{defn:multi-flow}
Given a Riemannian manifold $\mathcal M$ with boundary $\p \mathcal M$, let $A_1, \ldots, A_n$ be non-overlapping regions of $\partial\mathcal M$ (i.e.\ for $i\neq j$, $A_i\cap A_j$ is codimension-1 or higher in $\partial\mathcal M$) covering $\partial\mathcal M$ ($\cup_i A_i=\p \mathcal M $). A \emph{multiflow} is a set of vector fields $v_{ij}$ on $\mathcal M$ satisfying the following conditions:
\begin{align}
v_{ij}&=-v_{ji} \label{antisym}\\
\hat n \cdot v_{ij}
&= 0
\text{ on }A_k\quad(k\neq i,j)\label{noflux}\\
\nabla\cdot v_{ij}&=0 \label{divergenceless}\\
\sum_{i < j}^n |v_{ij}| &\leq \frac1{\GN}\ .\label{normbound}
\end{align}
\end{defn}

Given condition~\eqref{antisym}, there are $n(n-1)/2$ independent vector fields. Given condition~\eqref{noflux}, $v_{ij}$ has nonvanishing flux only on the regions $A_i$ and $A_j$, and, by~\eqref{divergenceless}, these fluxes obey
\begin{equation}
\int_{A_i}v_{ij}=-\int_{A_j}v_{ij}\ .
\end{equation}
Given conditions~\eqref{divergenceless} and~\eqref{normbound}, each $v_{ij}$ is a flow by itself. However, an even stronger condition follows: any linear combination of the form
\begin{equation}\label{combinations}
v=\sum_{i<j}^n \xi_{ij}v_{ij}\ ,
\end{equation}
where the coefficients $\xi_{ij}$ are constants in the interval $[-1,1]$, is divergenceless and, by the triangle inequality, obeys the norm bound $|v|\le1/\GN$, and is therefore also a flow.

Given a multiflow $\{v_{ij}\}$, we can define the $n$ vector fields
\begin{equation}
v_i:=\sum_{j=1}^n v_{ij}\ ,
\end{equation}
each of which, by the above argument, is itself a flow. Hence its flux on the region $A_i$ is bounded above by its entropy:
\begin{equation}\label{fluxbound}
\int_{A_i}v_i\le S(A_i)\ .
\end{equation}
The surprising statement is that the bounds~\eqref{fluxbound} are collectively tight. In other words, there exists a multiflow saturating all $n$ bounds~\eqref{fluxbound} simultaneously. We will call such a multiflow a \emph{max multiflow}, and its existence is our Theorem~\ref{thm17}:

\begin{theorem}[Max multiflow]\label{thm17}
There exists a multiflow $\{v_{ij}\}$ such that for each $i$, the sum
\begin{equation}\label{heresvi}
v_i:= \sum_{j=1}^n v_{ij}
\end{equation}
is a max flow for $A_i$, that is,
\begin{equation}\label{goal}
\int_{A_i}v_i=S(A_i)\ .
\end{equation}
\end{theorem}

Theorem 1 is a continuum version of a well-known theorem on multiflows on graphs, first formulated in~\cite{Kupershtokh1971} (although a correct proof wasn't given until~\cite{MR0437391,cherkassky1977}). However, the original graph-theoretic proof is discrete and combinatorial in nature and not easily adaptable to the continuum. Therefore, in Section~\ref{sec:proofs} we will give a continuum proof based on techniques from convex optimization. (This proof can be adapted back to the graph setting to give a proof there that is new as far as we know. We refer the reader to Section~\ref{sec:discrete}.) Furthermore, we emphasize that it should not be taken for granted that a statement that holds in the graph setting necessarily also holds on manifolds. In fact, we will give an example below of a graph-theoretic theorem concerning multiflows that is \emph{not} valid on manifolds.

A simple corollary of Theorem \ref{thm17} in the case $n=3$ is the nesting property for flows, mentioned in subsection \ref{sec:properties} above. This says that, given a decomposition of the boundary into regions $A,B,C$, there exists a flow $v(A,B)$ that is simultaneously a max flow for $A$ and for $AB$. 
In terms of the flows of Theorem \ref{thm17} (with $A_1=A$, $A_2=B$, $A_3=C$), such a flow is given by
\begin{equation}\label{vABdef}
v(A,B) = v_{AB}+ v_{AC}+v_{BC}\,.
\end{equation}
In \cite{Freedman:2016zud}, it was pointed out that, given a flow $v(A,B)$ and another one $v(B,A)$ which is a max flow for $B$ and $AB$, half of their difference
\begin{equation}
v(A:B):=\frac12(v(A,B)-v(B,A))\,,
\end{equation}
is a flow, and its flux gives half the mutual information:
\begin{equation}
\int_Av(A:B) = \frac12I(A:B)\,.
\end{equation}
Here, with $v(B,A)$ chosen similarly to \eqref{vABdef}, this flow is nothing but the component vector field $v_{AB}$ in the multiflow:
\begin{equation}
v(A:B) = v_{AB}\,.
\end{equation}
Furthermore, since $v_{AC}+v_{BC}$ is a max flow for $AB$, it must saturate on the RT surface $m(AB)=m(C)$, so $v_{AB}$ must vanish there. By the condition \eqref{noflux}, it also vanishes on $C$; therefore it can be set to zero in the homology region $r(C)$. With this choice, $v_{AB}$ is fully contained in the homology region $r(AB)$ (which is the complement of $r(C)$).

A more interesting corollary of Theorem \ref{thm17} is MMI.\footnote{While one may be tempted to similarly apply this theorem to $n$-party pure states for $n>4$ to potentially prove other holographic inequalities, such efforts have not been successful to date (but see footnote~\ref{locking footnote} on p.~\pageref{locking footnote}).}
Set $n=4$. Given a max multiflow $\{v_{ij}\}$, we construct the following flows:
\begin{align}
\begin{split}
	u_1 &:= v_{AC} + v_{AD} + v_{BC} + v_{BD} = \frac{1}{2}(v_A + v_B - v_C - v_D) \\
	u_2 &:= v_{AB} + v_{AD} + v_{CB} + v_{CD} = \frac{1}{2}(v_A - v_B +v_C - v_D) \\
	u_3 &:= v_{BA}+v_{BD}+v_{CA} + v_{CD} = \frac{1}{2}(-v_A + v_B + v_C - v_D)\ .
\end{split}
\end{align}
The second equality in each line follows from the condition~\eqref{antisym} and definition~\eqref{heresvi}.
Each $u_i$ is of the form~\eqref{combinations} and is therefore a flow, so its flux through any boundary region is bounded above by the entropy of that region. In particular,
\begin{align}\label{tworegions}
	S(AB) \geq \int_{AB} u_1\ , \quad S(AC) \geq \int_{AC} u_2\ , \quad S(BC) \geq \int_{BC}  u_3\ .
\end{align}
Summing these three inequalites and using~\eqref{goal} leads directly to MMI\@:
\begin{align}\label{MMIproved}
\begin{split}
	S(AB) + S(AC) + S(BC) &\ge \int_A(u_1+u_2)+\int_B(u_1+u_3)+\int_C(u_2+u_3) \\
&=	 \int_A v_A + \int_B v_B + \int_C  v_C - \int_{ABC}  v_D \\
	&= S(A)+S(B)+S(C)+S(D)\ .
\end{split}
\end{align}

The difference between the left- and right-hand sides of~\eqref{MMIproved} is $-I_3$, so (unless $-I_3$ happens to vanish) it is not possible for all of the inequalities~\eqref{tworegions} to be saturated for a given multiflow. However, Theorem~\ref{nmmf}, proved in Subsection~\ref{sec:nmmf}, shows as a special case that any single one of the inequalities~\eqref{tworegions} \emph{can} be saturated. For example, there exists a max multiflow such that
\begin{equation}\label{u1saturate}
\int_{AB}u_1 = S(AB)\ .
\end{equation}
In the graph setting, it can be shown that in fact any \emph{two} of the inequalities \eqref{tworegions} can be saturated \cite{Karzanov1978,Guyslain}; however, this is not in general true in the continuum.\footnote{We thank V. Hubeny for pointing this out to us. Futher details on this point will be presented elsewhere.}

\subsection{Threads}

\subsubsection{Theorem~\ref{thm17}}

We can also frame multiflows, Theorem~\ref{thm17}, and the proof of MMI in the language of bit threads. The concept of a multiflow is very natural from the viewpoint of the bit threads, since the whole set of flows $\{v_{ij}\}$ can be represented by a single thread configuration. Indeed, for each $v_{ij}$ ($i<j$) we can choose a set of threads with density equal to $|v_{ij}|$; given~\eqref{noflux}, these end only on $A_i$ or $A_j$. By~\eqref{Nfluxbound}, the number that connect $A_i$ to $A_j$ is at least the flux of $v_{ij}$:
\begin{equation}\label{ijbound}
N_{A_i A_j}\ge\int_{A_i}v_{ij}\ .
\end{equation}
Since the density of a union of sets of threads is the sum of their respective densities, by~\eqref{normbound} the union of these configurations over all $i<j$ is itself an allowed thread configuration. Note that this configuration may contain, in any given neighborhood, threads that are not parallel to each other, and even that intersect each other.

Now suppose that $\{v_{ij}\}$ is a max multiflow. Summing~\eqref{ijbound} over $j\neq i$ for fixed $i$ yields
\begin{equation}\label{N>flux}
\sum_{j\neq i}^n N_{A_i A_j}\ge\int_{A_i}v_i=S(A_i)\ .
\end{equation}
But, by~\eqref{Nbound}, the total number of threads connecting $A_i$ to all the other regions is also bounded \emph{above} by $S(A_i)$:
\begin{equation}\label{N<flux}
\sum_{j\neq i}^n N_{A_i A_j}\le S(A_i)\ .
\end{equation}
So the inequalities~\eqref{N>flux} and~\eqref{N<flux} must be saturated, and furthermore each inequality~\eqref{ijbound} must be individually saturated:
\begin{equation}
N_{A_i A_j} = \int_{A_i}v_{ij}\ .
\end{equation}
Thus, in the language of threads, Theorem~\ref{thm17} states that there exists a thread configuration such that, for all $i$,
\begin{equation}\label{MMFthread}
\sum_{j\neq i}^n N_{A_i A_j} = S(A_i)\ .
\end{equation}
We will call such a configuration a \emph{max thread configuration}.

We will now study the implications of the existence of a max thread configuration for three and four boundary regions.

\subsubsection{Three boundary regions}\label{sec:n=3}

For $n=3$, we have
\begin{equation}\label{n=3}
S(A) = N_{AB}+N_{AC}\ ,\qquad
S(B) = N_{AB}+N_{BC}\ ,\qquad
S(C) = N_{AC}+N_{BC}\ .
\end{equation}
Since $S(AB)=S(C)$, we find an elegant formula for the mutual information:
\begin{equation}\label{threadMI}
I(A:B) = 2N_{AB}\ .
\end{equation}
Thus, at least from the viewpoint of calculating the mutual information, it is as if each thread connecting $A$ and $B$ represents a Bell pair.\footnote{Strictly speaking, since a Bell pair has mutual information $2\ln2$, each thread represents $1/\ln2$ Bell pairs. If one really wanted each thread to represent one Bell pair, one could define the threads to have density $|v|/\ln2$, rather than $|v|$, for a given flow $v$.} Note that~\eqref{threadMI} also reestablishes the subadditivity property, since clearly the number of threads cannot be negative. Similarly to~\eqref{threadCE}, we also have for the conditional entropy,
\begin{equation}
H(B|A) = N_{BC}-N_{AB}\ .
\end{equation}

As mentioned in Subsection~\ref{sec:MMI}, the three mutual informations $I(A:B)$, $I(A:C)$, $I(B:C)$ determine the entropy vector in $\mathbb R^3$. Therefore, by~\eqref{threadMI} and its analogues, the thread counts $N_{AB}$, $N_{AC}$, $N_{BC}$ determine the entropy vector, and conversely are uniquely fixed by it. Thus, in the skeleton graph representation of the entropy vector shown in Fig.~\ref{fig:skeleton} (right side), we can simply put $N_{AB}$, $N_{AC}$, $N_{BC}$ as the capacities on the respective edges; in other words, the thread configuration ``is'' the skeleton graph.

\subsubsection{Four boundary regions}\label{n4}

For $n=4$, we have, similarly to~\eqref{n=3},
\begin{equation}\label{n=4}
\begin{split}
S(A) &= N_{AB}+N_{AC}+N_{AD}\\
S(B) &= N_{AB}+N_{BC}+N_{BD}\\
S(C) &= N_{AC}+N_{BC}+N_{CD}\\
S(D) &= N_{AD}+N_{BD}+N_{CD}\ .
\end{split}
\end{equation}
The entropies of pairs of regions, $S(AB)$, $S(AC)$, and $S(BC)$ also enter in MMI\@. A max thread configuration does not tell us these entropies, only the entropies of \emph{individual} regions. Nonetheless, for any valid thread configuration, we have the bound~\eqref{Nbound}. In particular, $S(AB)$ is bounded below by the total number of threads connecting $AB$ to $CD$:
\begin{equation}\label{ABbound}
S(AB) \ge N_{(AB)(CD)}=N_{AC}+N_{BC}+N_{AD}+N_{BD}\ .
\end{equation}
Similarly,
\begin{equation}\label{ACBCbound}
\begin{split}
S(AC) &\ge N_{AB}+N_{AD}+N_{BC}+N_{CD}\\
S(BC) &\ge N_{AB}+N_{AC}+N_{BD}+N_{CD}\ .
\end{split}
\end{equation}
Inequalities~\eqref{n=4}, \eqref{ABbound}, and~\eqref{ACBCbound} together  imply MMI.

\begin{figure}
\centering
\includegraphics[scale=0.75]{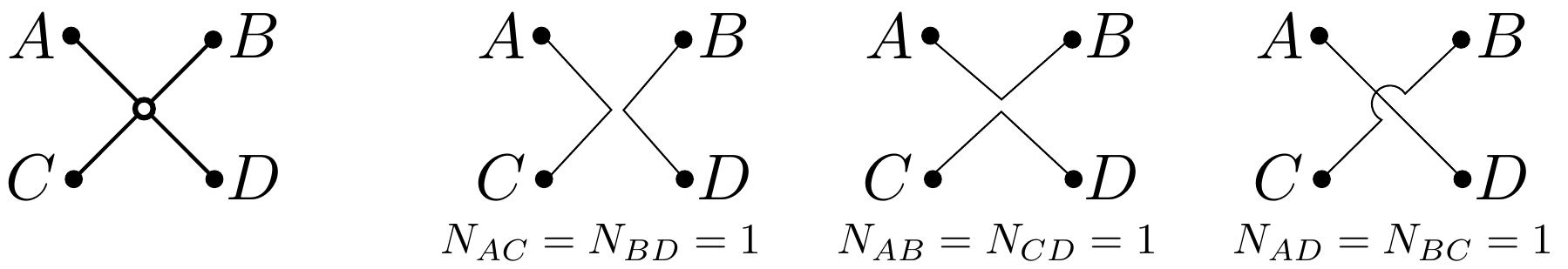}
\caption{
Left: Star graph with capacity 1 on each edge. The entropies derived from this graph are those of a perfect tensor, \eqref{PT_cond}, with $S_0=1$. In particular, all the pairwise mutual informations vanish and $-I_3=2$. Right: The three max thread configurations on this graph.
}\label{fig:xgraph}
\end{figure}

As we did for three parties, we can look at the mutual information between two regions. However, using~\eqref{n=4} and~\eqref{ABbound}, we now merely find a bound rather than an equality:
\begin{equation}\label{MIbound}
I(A:B) \le 2N_{AB}\ .
\end{equation}
Thus, in a four-party max configuration, each thread connecting $A$ and $B$ does not necessarily represent a ``Bell pair.'' To understand how this can occur, it is useful to look at a simple illustrative example, shown in Fig.~\ref{fig:xgraph}, which is a star graph where each edge has capacity 1. It is easy to evaluate the entropies of the single vertices and pairs. One finds that they have the form~\eqref{PT_cond} with $S_0=1$; in other words, this graph represents a perfect tensor. In particular, all pairwise mutual informations vanish, while $-I_3=2$. As shown in Fig.~\ref{fig:xgraph}, there are three max thread configurations. Each such configuration has two threads, which connect the external vertices in all possible ways.

The above example highlights the fact that, unlike for $n=3$, the thread counts $N_{A_i A_j}$ are not determined by the entropies: in~\eqref{n=4} there are only 4 equations for the 6 unknown $N_{A_i A_j}$, while the entropies of pairs of regions only impose the \emph{inequality} constraints~\eqref{ABbound}, \eqref{ACBCbound} on the $N_{A_i A_j}$. However, based on Theorem~\ref{nmmf} as summarized above (\ref{u1saturate}), we know that there exists a max thread configuration that saturates~\eqref{ABbound} and therefore~\eqref{MIbound}. The same configuration has $I(C:D)=2N_{CD}$. Similarly, there exists a (in general different) max configuration such that $I(A:C)=2N_{AC}$ and $I(B:D)=2N_{BD}$, and yet another one such that $I(A:D)=2N_{AD}$ and $I(B:C)=2N_{BC}$.
In summary, $\frac12I(A_i:A_j)$ is the minimal number of threads connecting $A_i$ and $A_j$, while~$-I_3$ is the total number of ``excess'' threads in any configuration:
\begin{equation}
-I_3 = \sum_{i<j}^n \left(N_{A_i A_j}-\frac12I(A_i:A_j)\right)\ .
\end{equation}
These $-I_3$ many threads are free to switch how they connect the different regions, in the manner of Fig.~\ref{fig:xgraph}.

So far we have treated the $n=3$ and $n=4$ cases separately, but they are related by the operation of merging boundary regions. For example, given the four regions $A$, $B$, $C$, $D$, we can consider $CD$ to be a single region, effectively giving a three-boundary decomposition. Under merging, not every four-party max thread configuration becomes a three-party max configuration. For example, in the case illustrated in Fig.~\ref{fig:xgraph}, if we consider $CD$ as a single region then, since $S(CD)=2$, any max thread configuration must have two threads connecting $CD$ to $AB$. Thus, the middle configuration, with $N_{AB}=N_{CD}=1$, is excluded as a three-party max configuration.

\section{State decomposition conjecture}\label{sec:statedecomp}

In this section we consider the thread configurations discussed in the previous section for different numbers of boundary regions. Taking seriously the idea that the threads represent entanglement in the field theory, we now ask what these configurations tell us about the entanglement structure of holographic states. We will consider in turn decomposing the boundary into two, three, four, and more regions.

For two complementary boundary regions $A$ and $B$, the number of threads connecting $A$ and $B$ in a max configuration is $N_{AB}=S(A)=S(B)$, so in some sense each thread represents an entangled pair of qubits with one qubit in $A$ and the other in $B$.
Of course, these qubits are not spatially localized in the field theory---in particular they are \emph{not} located at the endpoints of the thread---since even in a max configuration the threads have a large amount of freedom in where they attach to the boundary.

For three boundary regions, as discussed in Subsection~\ref{sec:n=3}, the max thread configuration forms a triangle, with the number of threads on the $AB$ edge fixed to be $N_{AB}=\frac12I(A:B)$ and similarly with the $AC$ and $BC$ edges. If we take this picture seriously as a representation of the entanglement structure of the state itself, it suggests that the state contains only bipartite entanglement. In other words, there is a decomposition of the $A,B,C$ Hilbert spaces,
\begin{equation}
\H_A = \H_{A_1}\otimes\H_{A_2}\ ,\qquad
\H_B = \H_{B_1}\otimes\H_{B_3}\ ,\qquad
\H_C = \H_{C_2}\otimes\H_{C_3}\ ,
\end{equation}
(again, this is \emph{not} a spatial decomposition) such that the full state decomposes into a product of three bipartite-entangled pure states:
\begin{equation}\label{threeregion}
\ket{\psi}_{ABC} = \ket{\psi_1}_{A_1B_1}\otimes\ket{\psi_2}_{A_2C_2}\otimes\ket{\psi_3}_{B_3C_3}\ ,
\end{equation}
each of which carries all the mutual information between the respective regions:
\begin{align}
S(A_1)&=S(B_1)=\frac12I(A:B)\nonumber \\
S(A_2)&=S(C_2)=\frac12I(A:C)\\
S(B_3)&=S(C_3)=\frac12I(B:C)\ .\nonumber
\end{align}
(Of course, any of the factors in \eqref{threeregion} can be trivial, if the corresponding mutual information vanishes.)

So far the picture only includes bipartite entanglement. For four boundary regions, however, with only bipartite entanglement we would necessarily have $-I_3=0$, which we know is not always the case in holographic states. Furthermore, as we saw in Subsection~\ref{n4}, even in a max thread configuration, the number of threads in each group, say $N_{AB}$, is not fixed. We saw that there is a minimal number $\frac12I(A_i:A_j)$ of threads connecting $A_i$ and $A_j$, plus a number $-I_3$ of ``floating'' threads that can switch which pairs of regions they connect. This situation is summarized by the skeleton graph of Fig.~\ref{fig:skeleton}, which includes six edges connecting pairs of external vertices with capacities equal to half the respective mutual informations, plus a star graph connecting all four at once with capacity $-\frac12I_3$. The star graph has perfect-tensor entropies. This suggests that the state itself consists of bipartite-entangled pure states connecting pairs of regions times a four-party perfect tensor:
\begin{equation}
\label{fourregion}
\begin{aligned}
\ket{\psi}_{ABCD} &= \ket{\psi_1}_{A_1B_1}\otimes
\ket{\psi_2}_{A_2C_2}\otimes
\ket{\psi_3}_{A_3D_3}\otimes
\ket{\psi_4}_{B_4C_4}\otimes
\ket{\psi_5}_{B_5D_5}\otimes
\ket{\psi_6}_{C_6D_6} \\
&\otimes \ket{PT}_{A_7B_7C_7D_7}\ .
\end{aligned}
\end{equation}
(Again, any of these factors can be trivial.) This is the simplest ansatz for a four-party pure state consistent with what we know about holographic entanglement entropies. In this ansatz the MMI property is manifest.

The conjectures \eqref{threeregion}, \eqref{fourregion} for the form of the state for three and four regions respectively are in fact equivalent. The four-region conjecture implies the three-region one, either by taking one of the regions to be empty or by merging two of the regions. A four-party perfect tensor, under merging two of the parties, factorizes into two bipartite-entangled states. For example, if we write $C'=CD$, then the bipartite-entangled factors in~\eqref{fourregion} clearly take the form~\eqref{threeregion}, while the perfect-tensor factor splits into bipartite-entangled pieces:
\begin{equation}
\ket{PT}_{A_7B_7C_7D_7} = \ket{\psi_1'}_{A_7C'_{7,1}}\otimes\ket{\psi'_2}_{B_7C'_{7,2}}
\end{equation}
for some decomposition $\mathcal H_{C'_7} = \mathcal H_{C_7} \otimes \mathcal H_{D_7} \cong \mathcal H_{C'_{7,1}} \otimes
 \mathcal H_{C'_{7,2}}$, as follows from the fact that ${I(A_7:B_7)}_{\ket{PT}}=0$.

Conversely, the three-region decomposition implies the four-region one, as follows \cite{Nezami:2016zni}. Suppose a pure state on $ABCD$ contains only bipartite entanglement when any two parties are merged. For example, when merging $C$ and $D$, it has only bipartite entanglement, and in particular contains an entangled pure state between part of $A$ and part of $B$, with entropy $I(A:B)/2$. There is a similar pure state shared between any pair of parties. These factors carry all of the pairwise mutual information, so what is left has vanishing pairwise mutual informations and is therefore a four-party PT (or, if $-I_3$ vanishes, a completely unentangled state).

We remind the reader that, throughout this paper, we have been working in the classical, or large-$N$, limit of the holographic system, and we emphasize that the state decomposition conjectures stated above should be understood in this sense.
Thus we are not claiming that the state takes the form~\eqref{threeregion} or~\eqref{fourregion} exactly, but rather only up to corrections that are subleading in $1/N$.
If we consider, for example, a case where $I(A:B)=0$ at leading order, such as where $A$ and $B$ are well-separated regions, the three-party decomposition~\eqref{threeregion} would suggest that $\rho_{AB}=\rho_A\otimes\rho_B$.
However, even in this case $I(A:B)$ could still be of order $O(1)$, so we should not expect this decomposition to hold approximately in any norm, but rather in a weaker sense.

Support for these conjectures comes from tensor-network toy models of holography~\cite{pastawski2015holographic,hayden2016holographic,Nezami:2016zni}. Specifically, it was shown in~\cite{Nezami:2016zni} that random stabilizer tensor network states at large $N$ indeed have the form~\eqref{threeregion}, \eqref{fourregion} at leading order in $1/N$.
More precisely, these decompositions hold provided one traces out $O(1)$ many degrees of freedom in each subsystem.
In other words, there are other types of entanglement present (such as GHZ-type entanglement), but these make a subleading contribution to the entropies. We believe that it would be interesting to prove or disprove the state decomposition conjectures~\eqref{threeregion}, \eqref{fourregion}, as well as to sharpen them by clarifying the possible form of the $1/N$ corrections.

Finally, we note that it is straightforward to generalize~\eqref{threeregion} and~\eqref{fourregion} to more than four regions. Namely, we can conjecture that for $n$ parties, $\ket{\psi}_{A_1\dots A_n}$ decomposes into a direct product of states, each realizing an extremal ray in the $n$-party holographic entropy cone. (Note that the above procedure of using the three-party decomposition to remove bipartite entanglement between any two parties also works in the $n$-party case, but what is left is not necessarily an extremal vector as it is for $n=4$.) A new feature that arises for $n>4$, as mentioned in Subsection~\ref{sec:entropy cones}, is that a generic vector in the holographic entropy cone no longer admits a unique decomposition into extremal rays. Therefore the amount of entropy carried by each factor in the state decomposition cannot be deduced just from the entropy vector, but would require some more fine-grained information about the state. Another new feature that arises for $n>5$ is that the extremal rays no longer arise only from perfect tensors; rather, new entanglement structures are involved. It would be interesting to explore whether the thread picture throws any light on these issues.

\section{Proofs}\label{sec:proofs}

In this section, we give proofs of our main results on the existence of multiflows in Riemannian geometries.
We are not claiming mathematical rigor, particularly when it comes to functional analytical aspects.
To simplify the notation, we set $\GN=1$ throughout this section.

\subsection{Theorem~\ref{thm17}}

For convenience, we repeat the definition of a multiflow and the statement of Theorem~\ref{thm17}.

\setcounter{defn}{0}

\begin{defn}[Multiflow]
Given a Riemannian manifold $\mathcal M$ with boundary\footnote{As mentioned in footnote \ref{internalboundary}, the case where $\mathcal{M}$ has an ``internal boundary'' $\mathcal{B}$ is also physically relevant. In this case, $\mathcal{B}$ not included in the decomposition into regions $A_i$, all flows are required to satisfy the boundary condition $n\cdot v=0$ on $\mathcal{B}$, and homology relations are imposed relative to $\mathcal{B}$. The reader can verify by following the proofs, with $\partial\mathcal{M}$ replaced by $\partial\mathcal{M}\setminus\mathcal{B}$, that all of our results hold in this case as well.} $\p \mathcal M$, let $A_1, \ldots, A_n$ be non-overlapping regions of $\partial\mathcal M$ (i.e.\ for $i\neq j$, $A_i\cap A_j$ is codimension-1 or higher in $\partial\mathcal M$) covering $\partial\mathcal M$ ($\cup_i A_i=\p \mathcal M $). A \emph{multiflow} is a set of vector fields $v_{ij}$ on $\mathcal M$ satisfying the following conditions:
\begin{align}
v_{ij}&=-v_{ji} \label{antisym2}\\
\hat n \cdot v_{ij}
&= 0
\text{ on }A_k\quad(k\neq i,j)\label{noflux2}\\
\nabla\cdot v_{ij}&=0 \label{divergenceless2}\\
\sum_{i < j}^n |v_{ij}| &\leq 1\label{normbound2}
\end{align}
\end{defn}

\setcounter{theorem}{0}

\begin{theorem}[Max multiflow]\label{thm172}
There exists a multiflow $\{v_{ij}\}$ such that for each $i$, the sum
\begin{equation}
v_i:= \sum_{j=1}^n v_{ij}
\end{equation}
is a max flow for $A_i$, that is,
\begin{equation}\label{goal2}
\int_{A_i}v_i=S(A_i)\ .
\end{equation}
\end{theorem}

Our proof of Theorem~\ref{thm17} will not be constructive. Rather, using tools from the theory of convex optimization, specifically strong duality of convex programs,\footnote{We refer the reader to~\cite{Boyd2004:CO:993483} for an excellent guide to this rich subject, but we also recommend~\cite{Headrick:2017ucz} for a short physicist-friendly introduction summarizing the concepts and results applied here.}  we will establish abstractly the existence of a multiflow obeying~\eqref{goal2}. The methods employed here will carry over with only small changes to the discrete case, as shown in Section~\ref{sec:discrete}.

\begin{proof}[Proof of Theorem~\ref{thm172}]
As discussed in Subsection~\ref{sec:multiflows}, for any multiflow, $v_i$ is a flow and therefore obeys
\begin{equation}\label{flowbound2}
\int_{A_i}v_i\le S(A_i)\ .
\end{equation}
What we will show is that there exists a multiflow such that
\begin{equation}\label{fluxsum}
\sum_{i=1}^n \int_{A_i}v_i \ge\sum_{i=1}^n S(A_i)\ .
\end{equation}
This immediately implies that~\eqref{flowbound2} is saturated for all $i$.

In order to prove the existence of a multiflow obeying~\eqref{fluxsum}, we will consider the problem of maximizing the left-hand side of~\eqref{fluxsum} over all multiflows as a convex optimization problem, or convex program. That this problem is convex follows from the following facts: (1) the variables (the vector fields $v_{ij}$) have a natural linear structure; (2) the equality constraints~\eqref{antisym2}, \eqref{noflux2}, \eqref{divergenceless2} are affine (in fact linear); (3) the inequality constraint~\eqref{normbound2} is convex (i.e.\ it is preserved by taking convex combinations); (4) the objective, the left-hand side of~\eqref{fluxsum}, is a concave (in fact linear) functional. We will find the Lagrangian dual of this problem, which is another convex program involving the constrained \emph{minimization} of a convex functional. We will show that the objective of the dual program is bounded below by the right-hand side of~\eqref{fluxsum}, and therefore its minimum $d^\star$ is bounded below:
\begin{equation}\label{dbound}
d^\star \ge \sum_{i=1}^n S(A_i)\ .
\end{equation}
We will then appeal to strong duality, which states that the maximum $p^\star$ of the original (primal) program equals the minimum of the dual,
\begin{equation}
p^\star = d^\star \ .
\end{equation}
We thus obtain
\begin{equation}\label{pbound}
p^\star \ge\sum_{i=1}^n S(A_i)\ ,
\end{equation}
showing that there is a multiflow obeying~\eqref{fluxsum}.

To summarize, we need to (a) derive the dual program and show that strong duality holds, and (b) show that its objective is bounded below by $\sum_{i=1}^n S(A_i)$. We will do these in turn. Many of the steps are similar to those in the proof of the Riemannian max flow-min cut theorem, described in~\cite{Headrick:2017ucz}; the reader who wishes to see the steps explained in more detail should consult that reference.

\paragraph{(a) Dualization:} The Lagrangian dual of a convex program is defined by introducing a Lagrange multiplier for each constraint and then integrating out the original (primal) variables, leaving a program written in terms of the Lagrange multipliers. More specifically, an inequality constraint is enforced by a Lagrange multiplier $\lambda$ which is itself subject to the inequality constraint $\lambda\ge0$. In integrating out the primal variables, the objective plus Lagrange multiplier terms (together called the \emph{Lagrangian}) is minimized or maximized \emph{without} enforcing the constraints. The resulting function of the Lagrange multipliers is the objective of the dual program. The requirement that the minimum or maximum of the Lagrangian is finite defines the constraints of the dual program (in addition to the constraints $\lambda\ge0$ mentioned above). If the primal is a minimization program then the dual is a maximization one and vice versa.

In fact it is not necessary to introduce a Lagrange multiplier for each constraint of the primal program. Some constraints can be kept \emph{implicit}, which means that no Lagrange multiplier is introduced and those constraints \emph{are} enforced when integrating out the primal variables.

Our task is to dualize the program of maximizing $\sum_{i=1}^n \int_{A_i}v_i$ over all multiflows, i.e.\ over sets $\{v_{ij}\}$ of vector fields obeying~\eqref{antisym2}--\eqref{normbound2}; as discussed above, this is a convex program. We will choose to keep~\eqref{antisym2} and~\eqref{noflux2} implicit. For the constraint~\eqref{divergenceless2}, we introduce a set of Lagrange multipliers $\psi_{ij}$ ($i<j$), each of which is a scalar field on $\mathcal M$. Note that $\psi_{ij}$ is only defined for $i<j$ since, given the implicit constraint~\eqref{antisym2}, the constraint~\eqref{divergenceless2} only needs to be imposed for $i<j$. For the inequality constraint~\eqref{normbound2} we introduce the Lagrange multiplier $\lambda$, which is also a scalar function on $\mathcal M$ and is subject to the constraint $\lambda\ge0$. The Lagrangian is
\begin{equation}\label{Lag1}
\begin{aligned}
L\left[\{v_{ij}\},\{\psi_{ij}\},\lambda\right]
&=
\sum_{i=1}^n \int_{A_i}\sqrt h\sum_{j=1}^n \hat n\cdot v_{ij}\\
&{}+
\int_{\mathcal M}\sqrt g\left[
\sum_{i<j}^n \psi_{ij}\nabla\cdot v_{ij}+\lambda\left(1-\sum_{i<j}^n |v_{ij}|\right)
\right]\ .
\end{aligned}
\end{equation}
Rewriting the first term slightly, integrating the divergence term by parts, and using the constraint~\eqref{noflux2}, the Lagrangian becomes
\begin{equation}\label{Lag2}
\begin{aligned}
L\left[\{v_{ij}\},\{\psi_{ij}\},\lambda\right]
&=
\sum_{i<j}^n\left[\int_{A_i}\sqrt h\,\hat n\cdot v_{ij}(\psi_{ij}+1)+\int_{A_j}\sqrt h\,\hat n\cdot v_{ij}(\psi_{ij}-1)\right]\\
&{}+\int_{\mathcal M}\sqrt g\left[
\lambda-\sum_{i<j}^n \left(v_{ij}\cdot\nabla\psi_{ij}+\lambda|v_{ij}|\right)
\right]\ .
\end{aligned}
\end{equation}

We now maximize the Lagrangian with respect to $v_{ij}$ (again, only imposing constraints~\eqref{antisym2}, \eqref{noflux2} but \emph{not}~\eqref{divergenceless2}, \eqref{normbound2}). The requirement that the maximum is finite leads to constraints on the dual variables $\{\psi_{ij}\}$, $\lambda$. Since the Lagrangian, as written in~\eqref{Lag2}, is ultralocal in $v_{ij}$, we can do the maximization pointwise. On the boundary, for a given $i<j$, at a point in $A_i$ or $A_j$, $\hat n\cdot v_{ij}$ can take any value.
Therefore, in order for the maximum to be finite, its coefficient must vanish, leading to the constraints
\begin{equation}
\psi_{ij} = -1\text{ on }A_i\ ,\qquad
\psi_{ij} = 1\text{ on }A_j\ .
\end{equation}
When those constraints are satisfied, the boundary term vanishes. In the bulk, the term $-(v_{ij}\cdot\nabla\psi_{ij}+\lambda|v_{ij}|)$ is unbounded above as a function of $v_{ij}$ unless
\begin{equation}\label{lambdaconstraint}
\lambda\ge|\nabla\psi_{ij}|\ ,
\end{equation}
in which case the maximum (at $v_{ij}=0$) vanishes. (As a result of~\eqref{lambdaconstraint}, the constraint $\lambda\ge0$ is automatically satisfied and can be dropped.) The only term left in the Lagrangian is $\int_{\mathcal M}\sqrt g\,\lambda$.

All in all, we are left with the following dual program:
\begin{multline}\label{dual}
\text{Minimize }\int_{\mathcal M}\sqrt g\,\lambda
\text{ with respect to }\{\psi_{ij}\},\lambda \\
\text{ subject to }
\lambda\ge|\nabla\psi_{ij}|\ ,\quad
\psi_{ij} = -1\text{ on }A_i\ ,\quad
\psi_{ij} = 1\text{ on }A_j\ ,
\end{multline}
where again, $\psi_{ij}$ is defined only for $i<j$.

Strong duality follows from the fact that Slater's condition is satisfied. Slater's condition states that there exists a value for the primal variables such that all equality constraints are satisfied and all inequality constraints are strictly satisfied (i.e.\ satisfied with $\le$ replaced by $<$). This is the case here: the configuration $v_{ij}=0$ satisfies all the equality constraints and strictly satisfies the norm bound~\eqref{normbound2}.

\paragraph{(b) Bound on dual objective:} It remains to show that, subject to the constraints in~\eqref{dual}, the objective is bounded below by $\sum_i S(A_i)$.

First, because $\psi_{ij} = -1$ on $A_i$ and $1$ on $A_j$, for any curve $\mathcal{C}$ from a point in $A_i$ to a point in $A_j$, we have
\begin{align}\label{step1}
	\int_{\mathcal{C}}ds\, \lambda \geq \int_{\mathcal{C}} ds\,|\nabla\psi_{ij}|\geq \int_{\mathcal{C}}ds\, \hat t \cdot \nabla\psi_{ij} = \int_{\mathcal{C}} d\psi_{ij} = 2\ ,
\end{align}
where $ds$ is the proper length element, $\hat t$ is the unit tangent vector, and in the second inequality we used the Cauchy-Schwarz inequality. Now, for each $i$, define the function $\phi_i(x)$ on $\mathcal M$ as the minimum of $\int_{\mathcal C}ds\,\lambda$ over any curve from $A_i$ to $x$:
\begin{align}\label{def phi_i}
	\phi_i(x) = \inf_{\mathcal{C}\text{ from}\atop A_i\text{ to }x}\int_{\mathcal{C}} ds\,\lambda \ .
\end{align}
By virtue of~\eqref{step1},
\begin{align}\label{step2}
	\phi_i(x) + \phi_j(x) \geq 2\qquad (i\neq j)\ .
\end{align}
Define the region $R_i$ as follows:
\begin{equation}\label{def R_i}
R_i:=\{x\in\mathcal M:\phi_i(x)<1\} \ .
\end{equation}
It follows from~\eqref{step2} that $R_i\cap R_j=\emptyset$ for $i \neq j$. Given that $\lambda\ge0$, this implies that the dual objective is bounded below by the sum of the integrals on the $R_i$s:
\begin{align}\label{ineq1}
	\int_{\mathcal M}\sqrt g\, \lambda \ge \sum_{i=1}^n \int_{R_i} \sqrt g\,\lambda \ .
\end{align}

Finally, we will show that each term in the sum on the right-hand side of~\eqref{ineq1} is bounded below by $S(A_i)$. Using Hamilton-Jacobi theory, where we treat $\int_{C}ds\, \lambda$ as the action, it is straightforward to show that $|\nabla\phi_i| = \lambda$, so this term can be written $\int_{R_i}\sqrt g|\nabla\phi_i|$. This integral in turn equals the average area of the level sets of $\phi_i$ for values between 0 and 1. Since $\phi_i=0$ on $A_i$ and $\ge2$ on $\overline{A_i}$, each level set is homologous to $A_i$ and so has area at least as large as that of the minimal surface $m(A_i)$. This is precisely $S(A_i)$, so the average is also at least $S(A_i)$. (The reasoning here is the same as used in the proof of the max flow-min cut theorem; see in particular Step 3 of Section 3.2 in~\cite{Headrick:2017ucz}.) This completes the proof.
\end{proof}

We end this subsection with two comments on the proof. The first is that the converse to~\eqref{step1} holds, in other words, given a function $\lambda$ on $\mathcal M$ such that $\int_C ds\,\lambda\ge2$ for any curve $\mathcal{C}$ connecting different boundary regions, there exist functions $\psi_{ij}$ satisfying the constraints of the dual program~\eqref{dual}. These can be constructed in terms of the functions $\phi_i$ and regions $R_i$ defined above:
\begin{equation}
\psi_{ij} = \begin{cases}\phi_i-1\text{ on }R_i\\ 1-\phi_j\text{ on }R_j\\ 0\text{ elsewhere}\end{cases}\ .
\end{equation}
Thus~\eqref{dual} is equivalent to the following program:
\begin{multline}\label{dual2}
\text{Minimize }\int_{\mathcal M}\sqrt g\,\lambda
\text{ with respect to }\lambda \\
\text{ subject to }
\int_{\mathcal{C}}ds\,\lambda\ge2\text{ for all $\mathcal{C}$ connecting different boundary regions.}
\end{multline}
This program is the continuum analogue of the ``metrics on graphs'' type of program that arises as duals of graph multiflow programs (see~\cite{Guyslain}).

Second, we know from Theorem~\ref{thm17} that the bound~\eqref{pbound} is saturated, implying that~\eqref{dbound} is saturated. In fact, it is straightforward to construct the minimizing configuration $\{\psi_{ij}\},\lambda$ which achieves this bound. Letting $m(A_i)$ be the minimal surface homologous to $A_i$ and $r(A_i)$ the corresponding homology region, we set
\begin{equation}
\psi_{ij} = \begin{cases}-1\text{ on }r(A_i)\\1\text{ on }r(A_j)\\0\text{ elsewhere}\end{cases}\ ,\qquad
\lambda = \sum_i\delta_{m(A_i)}\ ,
\end{equation}
where $\delta_{m(A_i)}$ is a delta-function supported on $m(A_i)$.

\subsection{Theorem~\ref{nmmf}}\label{sec:nmmf}

In this subsection we prove a theorem which establishes a sort of ``nesting property'' for multiflows. The theorem states that a multiflow exists that not only provides a max flow for each individual region $A_i$ but also for any given set of regions $s\subset \{A_i\}$. (The example $n=4$, $s=AB$ was considered in~\eqref{u1saturate}.) The corresponding flow $v_s$ is defined as the sum of the vector fields $v_{ij}$ from $s$ to $s^c$:
\begin{equation}
v_s := \sum_{A_i\in s\atop A_j\not\in s}v_{ij}\ .
\end{equation}
Being a max flow means
\begin{equation}
\int_s v_s = S(s)\ .
\end{equation}
For example, this was applied in Subsection~\ref{sec:multiflows} to the four-region case with $s=AB$.

\begin{theorem}[Nested max multiflow]\label{nmmf}
Given a composite boundary region~$s\subset \{A_i\}$, there exists a multiflow $\{v_{ij}\}$ such that for each $i$, the sum
\begin{equation}
v_i:= \sum_{j=1}^n v_{ij}
\end{equation}
is a max flow for $A_i$, that is,
\begin{equation}
\int_{A_i}v_i=S(A_i)\ ,
\end{equation}
and the sum
\begin{equation}
v_s := \sum_{A_i\in s\atop A_j\not\in s}v_{ij}\ .
\end{equation}
is a max flow for $s$, that is,
\begin{equation}
\int_s v_s = S(s)\ .
\end{equation}
\end{theorem}

\begin{proof}[Proof of Theorem~\ref{nmmf}]
The proof proceeds very similarly to that for Theorem~\ref{thm17}; we will only point out the differences. Since $v_s$ is automatically a flow, in addition to~\eqref{flowbound2} we have
\begin{equation}
\int_s v_s\le S(s)\ .
\end{equation}
Therefore, to prove the theorem, it suffices to show that there exists a multiflow such that (in place of~\eqref{fluxsum}),
\begin{equation}\label{fluxsum2}
\int_s v_s+\sum_{i=1}^n \int_{A_i}v_i\ge S(s)+\sum_{i=1}^n S(A_i)\ .
\end{equation}
For this purpose we dualize the program of maximizing the left-hand side of~\eqref{fluxsum2} over multiflows. Compared to the proof of Theorem~\ref{thm17}, this adds a term $\int_s v_s$ to the primal objective, and therefore to the Lagrangian~\eqref{Lag1}. This term can be written
\begin{equation}
\sum_{\substack{i<j\\ A_i\in s\\ A_j\not\in s}}\int_{A_i}\sqrt h\,\hat n\cdot v_{ij}-
\sum_{\substack{i<j\\ A_i\not\in s\\ A_j\in s}}\int_{A_j}\sqrt h\,\hat n\cdot v_{ij}\ .
\end{equation}
This term has the effect, after integrating out the $v_{ij}$s, of changing the boundary conditions for the dual variables $\psi_{ij}$. The dual program is now
\begin{equation}\label{dual3}
\begin{aligned}
&\text{Minimize }\int_{\mathcal M}\sqrt g\,\lambda
\text{ with respect to }\{\psi_{ij}\},\lambda \\
&\quad\text{ subject to }
\lambda\ge|\nabla\psi_{ij}|\ ,\\
&\qquad\qquad\left.\psi_{ij}\right|_{A_i} = \begin{cases}
-2\,,\quad A_i\in s,A_j\not\in s \\
-1\,,\quad\text{otherwise}
\end{cases} \\
&\qquad\qquad\left.\psi_{ij}\right|_{A_j} = \begin{cases}
2\,,\quad A_i\not\in s,A_j\in s \\
1\,,\quad\text{otherwise}
\end{cases}\ .
\end{aligned}
\end{equation}
This implies that the bound~\eqref{step1} on $\int_{\mathcal{C}}ds\,\lambda$ for a curve $\mathcal{C}$ from a point in $A_i$ to a point in $A_j$ becomes
\begin{equation}\label{distbound}
\int_{\mathcal{C}}ds\,\lambda\ge
\begin{cases}2\ ,\quad A_i,A_j\in s\text{ or }A_i,A_j\notin s \\
3\ ,\quad A_i\in s,A_j\not\in s\text{ or }A_i\not\in s,A_j\in s
\end{cases}\ .
\end{equation}
We now define the functions $\phi_i$ and regions $R_i$ as in the proof of Theorem~\ref{thm172}, and in addition the function $\phi_s$ and region $R_s$:
\begin{align}
\phi_s(x)&:=\min_{A_i\in s}\phi_i(x)=\min_{\mathcal{C}}\int_{\mathcal{C}}ds\,\lambda \\
R_s&:=\{x\in\mathcal M:1<\phi_s<2\} \ ,
\end{align}
where $\mathcal{C}$ is any curve from $s$ to $x$.
It follows from~\eqref{def phi_i}, \eqref{def R_i}, and~\eqref{distbound} that the regions $R_i$ do not intersect each other nor $R_s$. Therefore the objective in~\eqref{dual3} is bounded below by
\begin{equation}
\int_{R_s}\sqrt g\,\lambda+\sum_{i=1}^n\int_{R_i}\sqrt g\,\lambda \ .
\end{equation}
The integral over $R_i$ is bounded below by $S(A_i)$ by the same argument as in the proof of Theorem~\ref{thm172}, and the integral over $R_s$ is by $S(s)$ by a similar one: Again, $\lambda=|\nabla\phi_s|$, so $\int_{R_s}\sqrt g\,\lambda=\int_{R_s}\sqrt g\,|\nabla\phi_s|$, which in turn equals the average area of the level sets of $\phi_s$ for values between 1 and 2. Since $\phi_s$ is 0 on $s$ and $\ge3$ on $\bar s$, those level sets are homologous to $s$. Therefore their average area is at least the area of the minimal surface homologous to $s$, which is $S(s)$.
\end{proof}

\section{Multiflows on networks}\label{sec:discrete}

In this section we investigate multiflows on networks. This study can be thought of as a discrete analogue of the results in the previous sections, with the spacetime replaced by a weighted graph, the flows by graph flows, and the Ryu-Takayanagi surfaces by minimal cuts. Since the results in this section are stand-alone mathematical results, we will remain agnostic as to how (and if) a network is obtained from a spacetime: this could be done e.g.\ via the graph models of~\cite{Bao:2015bfa}, in some other way, or not at all.

There are two motivations for the study of multiflows on networks: (1) it could yield new insights into discrete models of gravity, but also (2) it can produce new mathematical results and conjectures in graph theory.

In this section we will report on several items:
\begin{enumerate}
\item In Section~\ref{secmultiflowproofdiscrete} we will give a convex optimization proof of the discrete analogue of Theorem~\ref{thm17} (Theorem~\ref{thm2} below). Although Theorem~\ref{thm2} has been proven before in the literature using combinatorial methods, this convex optimization proof is new to the best of our knowledge, and closely follows the proof in the continuum setup.
\item We will prove a decomposition of an arbitrary network with three boundary vertices into three subnetworks, such that each subnetwork computes precisely one of the three boundary mutual informations, and has zero value for the other two mutual informations. Furthermore, we will conjecture a decomposition of an arbitrary network with four boundary vertices into $6+1$ subnetworks, such that each of the six networks computes precisely one of the six pairwise mutual informations, and has vanishing mutual informations for the other five pairs of boundary vertices, as well as vanishing tripartite information. The remainder subnetwork has vanishing mutual informations and has tripartite information equal to that of the original network. The tripartite decomposition is the discrete analogue of the decomposition in Section~\ref{sec:n=3} in the continuous case, and the four-partite decomposition is a slight generalization of the decomposition in Section~\ref{n4}. We will also conjecture a decomposition of networks with arbitrary number of boundary vertices.
\item On networks with positive rational capacities, we will give a constructive combinatorial proof of the existence of a certain configuration of flows (in what we will call the flow extension lemma, Lemma~\ref{lem:trick}), which by itself is sufficient to establish the nonnegativity of~tripartite information. The result applies to not only undirected graphs, but also more generally to a certain class of directed graphs which we call \AB (to be defined later).
\end{enumerate}

\subsection{Background on networks}\label{subsec:background_definitions}
Denote a graph by $(V,E)$, where $V$ is the set of vertices and $E$ is the set of edges. We first consider the case of \emph{directed} graphs. For an edge $e \in E$, denote by $s(e)$ and $t(e)$ the source and target, respectively, of $e$. A capacity function on $(V,E)$ is a map $c:\, E \to \mathbb{R}_{\geq 0}$. For each $e \in E$, $c^e:= c(e)$ is called the capacity of $e$. We refer the graph $(V,E)$ together with a capacity function $c$ as a \emph{network} $\Sigma = (V,E,c)$.  Given a network $\Sigma$, we designate a subset of vertices $\partial \Sigma \subset V$ as the \emph{boundary} of $\Sigma$. Vertices in $\partial \Sigma$ play the role of sources or sinks.

\begin{defn}[Discrete flows]\label{flow_definition}
Given a network $\Sigma = (V,E,c)$, a \emph{flow} on $\Sigma$ is a function $v:\, E \to \mathbb R_{\geq 0}$ on edges such that the following two properties hold.
\begin{enumerate}
\item Capacity constraint:  for all edges $e\in E$, $|v^e| \leq c^e$.
\item Flow conservation: for all vertices $x \not\in \partial \Sigma $,
\begin{align}\label{flow_conservation}
	\sum_{E \ni e\colon t(e) = x} v^{e} = \sum_{E \ni e \colon s(e) = x} v^{e}\ .
\end{align}
\end{enumerate}
\end{defn}

For a network $\Sigma = (V,E,c)$, define the virtual edge set $\tilde{E}$ to be the set of edges obtained by reversing all directions of edges in $E$. Then clearly a flow $v$ on $\Sigma$ can be uniquely extended to a function, still denoted by $v$, on $E \sqcup \tilde{E}$ such that $v^{e} = -v^{\tilde{e}}$ where $\tilde{e} \in \tilde{E}$ corresponds to $e \in E$. All flows are implicitly assumed to have been extended in this way. Then the flow conservation property in \eqref{flow_conservation} can be rewritten as
\begin{align}
\label{flow_conservation2}
	\sum_{E\sqcup \tilde{E} \ni e\colon s(e) = x} v^{e} = 0,\quad \forall x \in V\ .
\end{align}

\begin{defn}
Let $v$ be a flow on $\Sigma = (V,E,c)$ and $A \subset \partial \Sigma$ be a subset of the boundary vertices.
\begin{enumerate}
\item The flux  of $v$ out of $A$ is defined to be
\begin{align}
S_{\Sigma}(A;v):= \sum_{E \sqcup \tilde{E} \ni e \colon s(e) \in A } v^{e}\ .
\end{align}
\item A max flow on $A$ is a flow that has maximal flux out of $A$ among all flows. The maximal flux is denoted by $S_{\Sigma}(A)$ (or $S(A)$ when no confusion arises), i.e.
\begin{align}
S_{\Sigma}(A) = \max_{v \text{ flow}} \,S_{\Sigma}(A;v)\ .
\end{align}
\item An edge cut set $C$ with respect to $A$ is a set of edges such that there exists a partition $V = V_1 \sqcup V_2$ with $A \subset V_1, \ \partial \Sigma \setminus A \subset V_2$, and $C = \{e \in E\,:\, s(e) \in V_1, t(e) \in V_2 \}$. The value of $C$ is defined to be $|C| = \sum_{e \in C} c^{e}$. A min cut with respect to $A$ is an edge cut set that has minimal value among all edge cut sets.
\end{enumerate}
\end{defn}

It is a classical result (the max-flow min-cut theorem)~\cite{ford1956maximal, elias1956note} that $S(A)$ is equal to the value of a min cut with respect to $A$.

In the continuum setup, $S(A)$ plays the role of entanglement entropy of the boundary region $A$ by the Ryu-Takayanagi formula. On networks, we still call $S(A)$ the ``entropy'' of $A$ by analogy. Similarly, for pairwise nonoverlapping boundary subsets $A,B,C$, define the mutual information and the tripartite information by
\begin{align}
I(A:B) &:= S(A) + S(B) - S(AB) \\
-I_3(A:B:C) &:= S(AB) + S(BC) + S(AC) - S(A) - S(B) - S(C) - S(ABC)\ .
\end{align}

If $(V,E)$ is an \emph{undirected} graph, it can be viewed as a directed graph $(V, D(E))$ where $D(E)$ is obtained by replacing each edge $e \in E$ with a pair of parallel oppositely-oriented edges $e_1,e_2$.
An undirected network $\Sigma = (V,E,c)$ can be viewed as a directed network $D(\Sigma) =(V,D(E),D(c))$, where ${D(c)}^{e_1} = {D(c)}^{e_2}:= c^e$.
We define a flow on $\Sigma$ to be a flow on $D(\Sigma)$.
From the viewpoint of computing the fluxes of flows, we can always assume that a flow $v$ on $D(\Sigma)$ satisfies the following property.
Namely, for each pair of parallel edges $(e_1,e_2)$, $v$ is positive on at most one of the two edges.
This allows us to define a flow on an undirected network $\Sigma$ without referring to $D(\Sigma)$.
Firstly, we arbitrarily fix a direction on each edge $e \in E$.
Then a flow on $\Sigma$ is a function $v: E \to \mathbb{R}$ that satisfies the same two conditions given in Definition~\ref{flow_definition}. The convention here is that if $v$ is negative on some $e \in E$ with the pre-fixed direction, then it means $v$ flows backwards along $e$ with the value $|v^{e}|$.
Also note that when computing the value of a cut on $\Sigma$, the edges of $E$ should be treated as undirected or bidirected, rather than directed with the prefixed orientation. The concepts of max flows and min cuts can be translated to undirected networks in a straightforward way, and the max-flow/min-cut theorem still holds.

We will consider undirected networks in Section~\ref{secmultiflowproofdiscrete} and~\ref{subsec:Network decomposition}, and directed networks in Section~\ref{subsec:tricklemma}.

\subsection{Discrete multiflow theorem and convex duality}\label{secmultiflowproofdiscrete}

In this subsection, all networks $\Sigma = (V,E,c)$ are undirected.
For simplicity, assume the underlying graph is simple and connected.
We also assume an arbitrary direction for each edge has been chosen.
For an edge $e$ with $x = s(e), y = t(e)$, we write $e$ as $\xy$, then $\langle yx \rangle$ denotes the edge $\tilde{e}$ in $\tilde{E}$.
The relation $x \sim y$ means $\xy \in E \sqcup \tilde{E}$, or equivalently, $x$ and $y$ are connected by an undirected edge.
For consistency with the notation used in the continuum setup, where different flows are labeled by subscripts, we will reserve subscripts of flows for the same purpose, and write $v^e$ for the value of a flow $v$ on an edge $e$.
Also, for convenience, we will label the vertices of $\partial \Sigma$ by $A_i$, $i\in \{1,2,\ldots,n\}$.
We do not allow for multiple boundary vertices to belong to the same $A_i$, but this will not result in any loss of generality.
Furthermore, we will assume, also without loss of generality, that each vertex $A_i\in\pd\Sigma$ connects to precisely one vertex $A_{\bar i}$ in $V \setminus \pd\Sigma$.

\begin{defn}[Multiflow, discrete version]
Given a network $\Sigma$ with boundary $\pd\Sigma$, a multiflow $\{ v_{ij} \}$ is a set of flows such that:
\begin{enumerate}
\item For all $i, j \in \{1,\dots,n\}$,
\begin{equation} \label{eqn:compat-disc-1}
v_{ij} = - v_{ji} \quad \text{and} \quad v_{ij}^{\langle k \bar{k} \rangle} = 0\ ,
\end{equation}
for any boundary vertex $A_k \in \p\Sigma \setminus \{A_i,A_j\}$.
\item The set of flows is collectively norm-bounded:
for all edges $e \in E$,
\begin{equation} \label{eqn:compat-disc-2}
\sum_{i < j}^n |v_{ij}^e | \leq c^e\ .
\end{equation}
\end{enumerate}
\end{defn}
As before, the compatibility condition ensures that any linear combination $\sum_{i < j}^n \xi_{ij} v_{ij}$ is a flow, provided that $|\xi_{ij}| \leq 1$.

The following theorem is well known in the multicommodity literature. It was first formulated in~\cite{Kupershtokh1971} with a correct proof given in~\cite{MR0437391,cherkassky1977} which is based on a careful analysis of the structure of max multiflows and involves delicate flow augmentations (cf.~\cite{frank1997integer}). Here we adapt the proof in the continuum setup to give a new proof (as far as we know) of this theorem via the convex optimization method. The proof proceeds mostly in parallel with the continuum case with some small changes. Readers who are not interested in the proof can skip to the next subsection.  Theorem~\ref{nmmf} also holds on networks, as it follows from the locking theorem~\cite{Karzanov1978}. A convex optimization proof of Theorem~\ref{nmmf} on networks would be similar to the continuum proof in Section~\ref{sec:proofs}. We will not, however, spell out the details here.

\begin{theorem}[Max multiflow, discrete version]\label{thm2}
Given a network $\Sigma$, there exists a multiflow $\{v_{ij}\}$ such that, for any $i$, the flow $\sum_{j=1}^n v_{ij}$ is a max flow on $A_i$.
\end{theorem}

\begin{proof}
Our goal is to determine the value
\begin{align}
\label{disc_primal}
p^\star =  \text{Maximize} \quad \sum_{i\neq j}^n v_{ij}^{\langle i \bar i \rangle}  \quad \text{subject to (\ref{flow_conservation2}), (\ref{eqn:compat-disc-1}), (\ref{eqn:compat-disc-2})}
\end{align}
of our primal optimization problem. We collect all constraints in an explicit manner  except the condition $v_{ij} = -v_{ji}$, which is treated as an implicit constraint. Introduce
\begin{align} \label{eqn:lagrangian-discr}
	L[v_{ij},\psi_{ij},\lambda] &= \sum_{i \not= j}^n v_{{i}{j}}^{\langle {i}\bar{i} \rangle} + \sum_{{i}<{j}}^n \sum_{y\in V\setminus\{i,j\}} \sum_{x\sim y} \psi_{{i}{j}}^{y} v_{{i}{j}}^{\xy}
	  +\sum_{\xy\in E} \lambda^\xy\lb c^{\xy} - \sum_{{i}<{j}}^n | v_{{i}{j}}^\xy | \rb\ ,
\end{align}
where $\psi_{ij}^x$ are Lagrange multipliers that enforce~\eqref{flow_conservation2} and~\eqref{eqn:compat-disc-1}, and $\lambda^{\la xy \ra}$ those that enforce~\eqref{eqn:compat-disc-2}. The Lagrange dual function is then
\begin{align}
	g(\psi_{ij},\lambda) = \sup_{v_{ij}} L[v_{ij},\psi_{ij},\lambda]\ .
\end{align}
As in the proof of Theorem~\ref{thm17}, we introduce $d^\star = \inf g(\psi_{ij},\lambda)$. By Slater's theorem, we again see that strong duality holds, as the primal constraints are strictly satisfied for the choice $v^\xy_{ij} = 0$. Thus, $p^\star = d^\star$, and we have reduced our primal objective~\eqref{disc_primal} to solving for $d^\star$.

By rewriting the $i\neq j$ sum in the first term of \eqref{eqn:lagrangian-discr} as two sums over $i<j$ and $j<i$, and interchanging the $i,j$ summation indices in the second sum, \eqref{eqn:lagrangian-discr} can be simplified to
\begin{equation}\label{eqn:L-simpl}
	L[v_{ij},\psi_{ij},\lambda] = \sum_{\xy\in E} \lsb \sum_{{i}<{j}}^n \lb \psi_{{i}{j}}^{y} - \psi_{{i}{j}}^{x} \rb v_{{i}{j}}^{\xy} + \lambda^\xy\lb  c^{\xy} - \sum_{{i}<{j}}^n | v_{{i}{j}}^\xy | \rb \rsb\ .
\end{equation}
Note that here we have introduced the boundary values
\begin{equation}\label{psibc}
-\psi_{{i}{j}}^{i} = \psi_{{i}{j}}^{j} = 1\ ,
\end{equation}
which should not be confused with the adjustable Lagrange multipliers $\psi_{ij}^x$. For fixed $\{\psi_{ij}^x\}$, we can always choose $v_{ij}^{\la xy\ra}$ so that $\operatorname{sgn}(v_{ij}^{\la xy\ra }) =\operatorname{sgn}(\psi_{ij}^y-\psi_{ij}^x)$.
Hence,
\begin{equation}
\begin{aligned}
	g(\psi_{ij},\lambda) &= \sup_{v_{ij}} \sum_{\xy\in E} \lsb \sum_{{i}<{j}}^n \left| \psi_{{i}{j}}^{y} - \psi_{{i}{j}}^{x} \right| | v_{{i}{j}}^{\xy} | + \lambda^\xy\lb c^{\xy} - \sum_{{i}<{j}}^n | v_{{i}{j}}^\xy | \rb \rsb \\
&= \sup_{v_{ij}} \sum_{\xy\in E} \lsb \sum_{{i}<{j}}^n \lb \left| \psi_{{i}{j}}^{y} - \psi_{{i}{j}}^{x} \right| - \lambda^\xy \rb | v_{{i}{j}}^{\xy} | + \lambda^\xy c^{\xy} \rsb\ .
\end{aligned}
\end{equation}
We observe that this is finite if and only if
\be
\left| \psi_{{i}{j}}^{y} - \psi_{{i}{j}}^{x} \right| - \lambda^\xy \leq 0\quad  \forall\ {i},{j} \in \p\Sigma\ , \quad \xy \in E\ ,
\ee
in which case the supremum is obtained by setting $v_{{i}{j}}^{\xy}=0$ everywhere. Therefore, it follows that
\be
\label{ref1}
d^\star = \inf g(\psi_{ij},\lambda) = \inf_\lambda \sum_{\xy\in E} \lambda^\xy c^{\xy}\ ,
\ee
subject to the edgewise condition
\be
\label{ref2}
\lambda^\xy \geq \max_{{i},{j} \in\p\Sigma} \left| \psi_{{i}{j}}^{y} - \psi_{{i}{j}}^{x} \right|\ .
\ee
This can be compactly written as
\be
d^\star = \min_{\psi_{ij}} \sum_{\xy\in E} c^\xy \max_{{i},{j}\in\p\Sigma} \left| \psi_{{i}{j}}^{x} - \psi_{{i}{j}}^{y} \right|\ ,
\ee
where the minimization over $\psi_{ij}$ is only subjected to the boundary condition~\eqref{psibc}.

We would like to prove that $p^{\star} = d^{\star} = \sum_{i=1}^n S(A_i)$. By the definition of $p^{\star}$, we have $p^{\star} \leq \sum_{i=1}^{n}S(A_i)$. Referring to~\eqref{psibc}, \eqref{ref1}, and~\eqref{ref2}, it suffices to prove that given the constraints
\be
\label{eq216}
	-\psi_{{i}{j}}^{i} =  \psi_{{i}{j}}^{j} = 1\ ,\quad \lambda^\xy \geq \left| \psi_{{i}{j}}^{y} - \psi_{{i}{j}}^{x} \right| \quad\text{for all $i,j\in\p\Sigma$}\ ,
\ee
we have for all $\lambda$
\be\label{dual_goal}
\sum_{\xy\in E} \lambda^\xy c^{\xy} \geq \sum_{i=1}^n S(A_i)\ .
\ee
This is done in the following steps.

First, let $C_{ij}$ be any (\emph{undirected}) path from a vertex $i$ to another vertex $j$. Utilizing~\eqref{eq216}, we see that for the particular case where $i,j \in \p\Sigma$,
\be
\label{equals2}
\sum_{\xy \in C_{{i}{j}}} \lambda^\xy \geq \sum_{\xy \in C_{{i}{j}}} \left| \psi_{{i}{j}}^{y} - \psi_{{i}{j}}^{x} \right| \geq \sum_{\xy \in C_{{i}{j}}} \lb \psi_{{i}{j}}^{y} - \psi_{{i}{j}}^{x} \rb =  \psi^{j}_{{i}{j}} - \psi^{i}_{{i}{j}} = 2\ .
\ee
Now, for any boundary vertex ${i} \in \p\Sigma$, we define the function $f_i$ on the set of vertices by
\begin{align}
	f_{i}(x)= \inf_{C_{ix}}\sum_{\langle uv\rangle \in C_{ix}} \lambda^{\langle uv \rangle}\ .
\end{align}
It immediately follows from~\eqref{equals2} that for any ${i}, {j}\in\p\Sigma$,
\be
\label{sum35}
f_{i}(x) + f_{j}(x) \geq 2\ .
\ee
Furthermore, we observe that by construction
\begin{align}
\begin{split}
	f_i(x) &\leq f_i(y)+\lambda^\xy\ ,
\end{split}
\end{align}
as the left-hand side chooses the path that minimizes the sum of $\lambda^{\la uv \ra}$ over paths from $i$ to $x$, whereas the right-hand side is the sum of $\lambda^{\la uv\ra}$ over a particular path from $i$ to $x$. Likewise, we can also interchange $x$ and $y$ to obtain another inequality (note $\lambda^\xy = \lambda^{\la yx\ra}$ by definition). These two inequalities combine to form the analog of the Hamilton-Jacobi equation from the proof of Theorem~\ref{thm17}:
\begin{equation}\label{absfxy}
| f_{i}(x) - f_{i}(y) | \leq \lambda^\xy\ .
\end{equation}

Finally, similar to the continuous case, we define $R_{i} \subset V$ to be the set of vertices for which $f_{{i}}(x)<1$. Note that given $i,j \in \p\Sigma$ with $i \not=j$, $R_i$ and $R_j$ do not overlap. To prove this, suppose there is a vertex $v \in R_i \cap R_j$. Then
\begin{align}
	f_i(v) + f_j(v) < 2\ ,
\end{align}
contradicting~\eqref{sum35}. If we define the function $g_i:\,  V \to \mathbb R$ such that
\begin{equation}
g_{{i}}(x) =
\begin{cases}
g_{{i}}(x) & x \in R_{{i}} \\
1             & x \in V \setminus R_{{i}}
\end{cases}\ ,
\end{equation}
it is readily apparent that $g_i(j) = 1 - \delta_{i,j}$ for all $j \in \p\Sigma$. We then claim the following inequality holds for $g_i$:
\begin{equation}
\label{equ:f'}
\sum\limits_{{i}=1}^n |g_{{i}}(x) - g_{{i}}(y)| \leq \lambda^{\xy}, \quad \forall \xy \in E\ .
\end{equation}
If this is indeed true, then we have successfully proven~\eqref{dual_goal}:
\begin{equation}
\sum\limits_{\xy \in E} c^{\xy}\lambda_{\xy} \geq \sum \limits_{{i}=1}^n \sum \limits_{\xy \in E} c^{\xy} |g_{{i}}(x) - g_{{i}}(y)|\geq \sum \limits_{{i}=1}^n S(A_i)\ ,
\end{equation}
where the last inequality follows from the same reasoning as the continuous case (with the details again spelled out in Step 3 of Section 3.2 of~\cite{Headrick:2017ucz}).  Alternately, one can write out the convex optimization program for maximizing the flux out of $A_i$ and see that the dual program is given by
\begin{align}
\inf_{g}\sum\limits_{\xy \in E} c^{\xy} |g(x) - g(y)| \quad \text{subject to}\quad g(j) = 1-\delta_{i,j},\, j \in \partial\Sigma\ .
\end{align}

Thus, it only remains for us to verify~\eqref{equ:f'}. For any fixed $\xy \in E$, consider the following cases. If $x,y \in R_i$ for some $i$, then the only term that contributes to the sum in~\eqref{equ:f'} is the $i$th term, and~\eqref{equ:f'} is true by~\eqref{absfxy}. If $x \in R_i$ but $y \notin R_i$, then there are two possibilities. If $y \in R_j$ for some $j \in \p\Sigma\setminus \{i\}$, then~\eqref{equ:f'} becomes
\begin{align}
	(1-f_i(x)) + (1 - f_j(x)) \leq \lambda^{\xy}\ ,
\end{align}
which is true by~\eqref{equals2}. If $y \notin R_j$ for any $j \in \p\Sigma$, then~\eqref{equ:f'} becomes
\begin{align}
	1 - f_i(x) \leq \lambda^{\xy}\ ,
\end{align}
which is true since $\lambda^\xy + f_i(x) \geq f_i(y) \geq 1$ by construction. Lastly, if $x,y \notin R_i$ for any $i \in\p\Sigma$, then the left-hand side of~\eqref{equ:f'} vanishes and is hence trivially true.
\end{proof}

\subsection{Network decomposition}\label{subsec:Network decomposition}

We now discuss network decomposition. Again let $\Sigma = (V,E,c)$ be an undirected network whose boundary $\partial \Sigma$ consists of $n$ components $A_1, \ldots, A_n$. The following definitions will be useful.

\begin{defn} 
\begin{enumerate}
\item The \emph{entropy ray} $R(\Sigma)$ associated with $\Sigma$ is the ray generated by the \emph{entropy vector}
\begin{equation}
\left( S(A_1), \dots, S(A_1A_2),\dots,S(A_1,\dots,A_n)  \right)\  ,
\end{equation}
where each entry $S(A_{i_1}\cdots A_{i_k})$ is the entropy of the boundary set~$A_{i_1}\cdots A_{i_k}$, i.e. the maximal flux out of $A_{i_1}\cdots A_{i_k}$.
\item A network (or geometry) $\Sigma$ \emph{realizes} a ray $R$ if $R(\Sigma) = R$.
\end{enumerate}
\end{defn}

Note that if a vector $R$ is realized by a network (geometry), then any positive scalar multiple of $R$ is also realized by the same network (geometry) up to scaling the capacities (metric).

We will sometimes say that a network is of a certain type (Bell pair, perfect tensor, etc.) if it realizes an entropy ray of that particular type of entanglement.

\begin{defn}[Subnetworks and subnetwork decomposition]
A subnetwork of $\Sigma = (V,E,c)$ is a network $(V,E,c_1)$ such that $c_1 \leq c$ on all edges. We say $\Sigma$ decomposes into subnetworks $\Sigma_{1},\dots,\Sigma_m$ for $\Sigma_i = (V,E,c_i)$ if
\be
\sum_{i=1}^m c_{i}^e = c^e\ , \quad \forall e \in E\ .
\ee
\end{defn}

Two vertices in a network $\Sigma$ are \emph{connected by} $\Sigma$ if there is a path of nonzero capacities between them.

\begin{theorem}[Tripartite network decomposition]\label{thm4}
An arbitrary network $\Sigma$ with three boundary vertices $A_1, A_2, A_3$ decomposes into three subnetworks $\Sigma_{ij}$, $1 \leq i < j \leq 3$, such that
\begin{align}
\label{equ:subnetwork decomposition}
S_{\Sigma_{ij}}(A_i) = S_{\Sigma_{ij}}(A_j) = \frac{1}{2}I(A_i:A_j)\ , \quad \text{and} \quad S_{\Sigma_{ij}}(A_k) = 0\ , \quad k \neq i,j\ ,
\end{align}
where $I(A_i:A_j) = S_{\Sigma}(A_i) + S_{\Sigma}(A_j) - S_{\Sigma}(A_i A_j)$ is the mutual information between $A_i$ and $A_j$ on $\Sigma$. In particular, this implies that $\Sigma_{ij}$ connects only $A_i$ and $A_j$, that $\sum_{i<j} S_{\Sigma_{ij}}(A_k) = S_\Sigma(A_k)$ for every $k$,
and that the $A_i:A_j$ mutual information computed on $\Sigma_{ij}$ equals that computed on~$\Sigma$.
\end{theorem}

\begin{proof}
From Theorem~\ref{thm2}, there exists a multiflow $\{v_{ij}\,:\, 1 \leq i \neq j \leq 3\}$ such that $v_{ij} = -v_{ji}$ and for each $i$, $\sum_{j \neq i} v_{ij}$ is a max flow on $A_i$. Hence, we have
\begin{align}
\label{somemaxflows}
\begin{split}
S(A_1) &= S(A_1; v_{12}) + S(A_1; v_{13})\ ,\\
S(A_2) &= S(A_1; v_{12}) + S(A_2; v_{23})\ ,\\
S(A_3) &= S(A_1; v_{13}) + S(A_2; v_{23})\ .
\end{split}
\end{align}
It follows that
\begin{align}
S(A_i; v_{ij}) = \frac{1}{2}I(A_i:A_j)\ ,\quad i < j\ .
\end{align}
Also, by definition, we have
\begin{align}
S(A_k; v_{ij}) = 0\ , \quad \text{and} \quad S(A_j \; v_{ij}) = -S(A_i ; v_{ij})\ , \quad k \neq i,j\ .
\end{align}

For $i<j$, define the subnetwork $\Sigma_{ij} = (V,E,c_{ij})$ such that $c_{ij}^e = |v_{ij}^e|$ for all $e \in E$. By noting that $v_{ij}$ can be viewed as a fully saturated flow on $\Sigma_{ij}$, it is clear that the subnetworks $\Sigma_{ij}$ satisfy the condition in \eqref{equ:subnetwork decomposition}. However, the three subnetworks $\Sigma_{ij}$ do not necessarily add up to $\Sigma$. There may still be a residual capacity $c^e - \sum_{i<j} c^e_{ij}$ on each edge $e$ of $\Sigma$. To solve the issue, we simply append this capacity to one of the $\Sigma_{ij}$, say $\Sigma_{12}$. This does not change the maximal fluxes on $\Sigma_{12}$ due to \eqref{somemaxflows}. For instance, if $S_{\Sigma_{12}}(A_1) > I(A_1:A_2)/2$, then
\begin{align}
S(A_1) \geq S_{\Sigma_{12}}(A_1) + S_{\Sigma_{13}}(A_1) > S(A_1)\ ,
\end{align}
which is absurd.
\end{proof}

\begin{conj}[Four-partite network decomposition]\label{discconjn4}
An arbitrary network $\Sigma$ with four boundary vertices $A_1, \ldots, A_4$ decomposes into six pairwise subnetworks $\Sigma_{ij}$, $1 \leq i<j\leq 4$, together with a remainder subnetwork $\Sigma_r$. The subnetworks obey the following properties:
\begin{enumerate}
\item $\Sigma_{ij}$ connects only $A_i$ and $A_j$, and has fluxes $I(A_i:A_j)/2$ on $A_i$. In other words,
\begin{align}
S_{\Sigma_{ij}}(A_i) = S_{\Sigma_{ij}}(A_j) = \frac{1}{2}I(A_i:A_j)\ , \quad \text{and} \quad S_{\Sigma_{ij}}(A_k) = 0\ , \quad k \neq i,j\ .
\end{align}
The right-hand side condition implies that $S_{\Sigma_{ij}}(A_i A_j) = 0$, so the $A_i:A_j$ mutual information computed on $\Sigma_{ij}$ equals that computed on~$\Sigma$.
Moreover, the tripartite information calculated on $\Sigma_{ij}$ vanishes.
\item $\Sigma_r$ is a four-partite perfect tensor network with the same tripartite information as $\Sigma$.
In other words, all pairwise mutual informations vanish and
\begin{equation}
S_{\Sigma_r}(A_i) = \frac{1}{2}S_{\Sigma_r}(A_i A_j) = \frac{-I_3}{2}\ ,
\end{equation}
where $-I_3 = -I_3(A_1:A_2:A_3)$ is the tripartite information calculated on $\Sigma$.
\end{enumerate}
Properties 1 and 2 together imply that the subsystem entropies of the seven subnetworks add up to those of $\Sigma$, i.e.
$\sum_{i<j} S_{\Sigma_{ij}}(s) + S_{\Sigma_r}(s) = S_\Sigma(s)$ for every $s \subset \{A_i\}$.
\end{conj}
Although Theorem~\ref{thm2} is not sufficient to prove Conjecture~\ref{discconjn4}, we have numerical evidence for it in the form of direct computations for some network examples. Furthermore, it is in fact not difficult to state a decomposition conjecture for arbitrary number of boundary vertices.

\begin{conj}[Arbitrary $n$ network decomposition]\label{nconedeconj}
A network with $n\geq 2$ boundary vertices decomposes into subnetworks each of which are realizations of extremal rays of the $n$-boundary region holographic entropy cone.
\end{conj}

Note that the cone for $n$ boundary regions has $S_n$ permutation symmetry, and we are not modding out by this symmetry in Conjecture~\ref{nconedeconj}.

One part of the conjecture is immediate: the set of allowed subnetworks must contain realizations of all the extremal rays.
The difficulty comes from the converse, which is showing that the decomposition of an arbitrary network $\Sigma$ requires no other subnetworks beyond those in Conjecture~\ref{nconedeconj}.

For $n=5$, the holographic entropy cone is known~\cite{Bao:2015bfa}. In this case, a network $\Sigma$ should decompose into (at most) $10$~Bell pair subnetworks, $5$~four-partite perfect tensor networks, and $5$~five-partite perfect tensor networks. However, beyond five boundary regions, Bell pair and perfect tensor networks will no longer suffice, as it is known that the higher holographic entropy cones have extremal rays that only admit network realizations of nontrivial topology.

\subsection{\trick}\label{subsec:tricklemma}

We give another proof that $-I_3\ge0$ for discrete networks. The result holds not only for undirected networks, but also more generally for directed networks with rational capacities satisfying a property called {\it \AB}, as defined below. The proof here involves only properties of flows without referring to min-cuts. Thus it is interesting to see if the techniques used here can be generalized to obtain entropy inequalities involving more than four boundary regions~\cite{Bao:2015bfa}.

Except in Corollary~\ref{cor:general_I3} at the end of this subsection, we consider \emph{directed} networks $\Sigma = (V,E,c)$ with boundary $\partial \Sigma$ where the capacity function $c$ is the constant function that assigns $1$ to all edges. Hence, $c$ will often be suppressed from the notation below. Also, we will only consider flows which have values zero or one on all edges. By the Ford-Fulkerson algorithm, a max flow on a boundary subset $A \subset \partial \Sigma$ can always be taken to have such a property. A flow is essentially a collection of edges such that at each non-boundary vertex the collection contains an equal number of incoming edges and outgoing edges.

Let $v$ be any flow on $\Sigma$ and $A \subset \partial \Sigma$ be a boundary subset. Denote by $A^c$ the complement of $A$ in $\partial \Sigma$, and set $n = S_{\Sigma}(A;v)$.  It is not hard to show that $v$ contains $n$ edge-disjoint paths from $A$ to $A^c$.  Conversely, an arbitrary collection of $n$ edge-disjoint paths from $A$ to $A^c$ defines a flow whose flux out of $A$ is $n$. Consequently, the max flow $S(A)$ is equal to the maximum number of edge-disjoint paths from $A$ to $A^c$.

Given a flow $v$ on a network $\Sigma = (V,E)$, the residual network $\Res(v;\Sigma)$ is obtained from $(V,E)$ by reversing the direction of all edges on which $v$ takes non-zero value. By our convention, each edge in $\Res(v;\Sigma)$ still has capacity one.

\begin{lemma}\label{lem:Res}
Let $\Sigma, A, $ and $v$ be as above, then $S_{\Sigma}(A) = S_{\Res(v;\Sigma)}(A) + S_{\Sigma}(A;v)$.
\end{lemma}
\begin{proof}
Set $\tilde{\Sigma} = \Res(v;\Sigma)$. For each edge $e$ in $\Sigma$, denote by $\tilde{e}$ the corresponding edge in~$\tilde{\Sigma}$. Hence $\tilde{e}$ has the same direction as $e$ if and only if $v^e = 0$. Let $w$ be any flow on~$\tilde{\Sigma}$. We define a  flow $u$ on $\Sigma$ by adding $w$ to $v$, taking the direction of the flows into consideration. Explicitly, set
\begin{align}
u^e &= |v^e - w^{\tilde{e}}|\ .
\end{align}
That is, there is one unit of flow for $u$ on an edge if and only if either $v$ or $w$, but not both, occupies that edge. (When they both occupy an edge, the net result adding them cancels each other.) It can be shown that $u$ is a valid flow on $\Sigma$ and moreover,
\begin{align}
S_{\Sigma}(A;u) &=  S_{\tilde{\Sigma}}(A; w) + S_{\Sigma}(A;v)\ .
\end{align}
By taking $w$ to be a max flow on $A$, we obtain the $\geq$ direction in the lemma,
\begin{align}
S_{\Sigma}(A) &\geq S_{\tilde{\Sigma}}(A) + S_{\Sigma}(A;v)\ .
\end{align}

To prove the other direction, define $\tilde{v}$ to be the flow on $\tilde{\Sigma}$ such that $\tilde{v}^{\tilde{e}} = v^e$. Then it follows that
\begin{align}
S_{\tilde{\Sigma}}(A;\tilde{v}) = - S_{\Sigma}(A;v) \quad\text{and}\quad \Res(\tilde{v};\tilde{\Sigma}) = \Sigma\ .
\end{align}
Hence, by what we have showed above,
\begin{align}
S_{\tilde{\Sigma}}(A) &\geq S_{\Sigma}(A) -S_{\Sigma}(A;v)\ .
\end{align}
\end{proof}

\begin{defn}
A flow $v$ on $\Sigma$ is \emph{reachable} from a boundary set $A \subset \partial \Sigma$ if for any edge $e$ with $v^e \neq 0$, there exists a path contained in $v$ connecting a vertex in $A$ to $s(e)$. 
\end{defn}

\begin{defn}
A network $\Sigma$ is \emph{\AB} if at each non-boundary vertex, the total capacity of incoming edges is no greater than that of outgoing edges.
\end{defn}

If $v$ is a flow on an \AB network $\Sigma$, then $\Res(f;\Sigma)$ is also \AB.

\begin{lemma}[\trick]\label{lem:trick}
Let $\Sigma$ be an \AB graph with a partition $\partial \Sigma = B_1 \sqcup B_2 \sqcup B_3$ such that $S(B_1) = 0$.%
\footnote{Actually, a weaker condition is sufficient, namely, there being no flow from $B_1$ to $B_2$.} Then given any flow $v$ reachable from $B_1$, there exists a flow $\tilde{v}$ extending $v$ such that $\tilde{v}-v$ is a max flow on $B_2$. Furthermore, one can choose $\tilde{v}$ such that $\tilde{v}-v$ is a collection of $S(B_2)$ edge-disjoint paths from $B_2$ to $B_2^c$, and in particular, $\tilde{v}$ is reachable from $B_1 \sqcup B_2$. See Figure~\ref{fig:trick_lemma} (Left) for a schematic picture of $\tilde{v}$.
\end{lemma}
\begin{figure}
\centering
\includegraphics[scale=1]{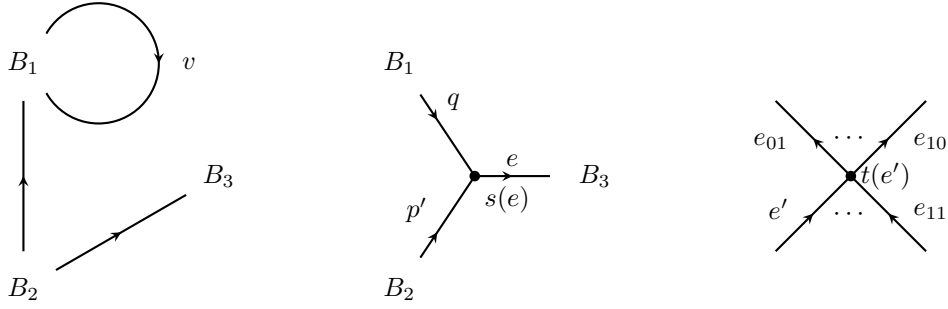}
\caption{(Left) A flow configuration of $\tilde{v}$ in Lemma~\ref{lem:trick}; (Middle) A configuration near $s(e)$ when $e$ is a type $(1,1)$ edge; (Right) A local picture at $t(e')$ when there are no outgoing edges with type $(0,0)$.}\label{fig:trick_lemma}
\end{figure}
\begin{proof}
Let $w$ be a max flow on $B_2$ consisting of $S(B_2)$ edge-disjoint paths from $B_2$ to $B_2^c$. We call an edge $e$ of type $(i,j)$, if $v^e = i, w^e = j$, $i,j = 0,1$. Let $n$ be the number of type $(1,1)$ edges. We construct a flow $\tilde{v}$ satisfying the requirement in the statement of the lemma by induction on $n$. If $n=0$, then $\tilde{v}:= v + w$ is such a flow.

If $n>0$, pick a path $p$ of $w$ such that $e$ is the first edge of $p$ along its direction that has type $(1,1)$. Truncate the path $p$ at the vertex $s(e)$, only keeping the first half from its initial vertex to $s(e)$ and dropping the second half from $w$ (so $w$ temporarily becomes an invalid flow). Update the type of these edges in the second half by subtracting $(0,1)$ from the original types. Denote the first half of $p$, i.e. the remaining part of $p$, by $p'$. Since $v$ is reachable from $B_1$, there exists a path $q$ of $v$ from $B_1$ to $s(e)$. Note that $s(e)$ cannot be a vertex of $B_2$ since otherwise there would be a flow from $B_1$ to $B_2$, contradicting to our assumption. Hence, $p'$ is not empty. See Figure~\ref{fig:trick_lemma} (Middle). Note that at the moment the number of type $(1,1)$ edges is at most $n-1$.

We now proceed in an algorithmic approach.

\vspace{0.3cm}
\noindent {\bf (ENTER):} If $p'$ is a complete path, i.e. it ends at a boundary vertex, {\bf goto} {\bf (EXIT)}. Otherwise, continue in the following two exclusive cases.

\vspace{0.3cm}
\noindent {\bf Case \boldmath{$I$}: there is an outgoing edge of type \boldmath{$(0,0)$} at the end point of \boldmath{$p'$}}. We simply append any outgoing edge of type $(0,0)$ to $p'$ and change its type to $(0,1)$. The augmented path is still denoted by $p'$, and we always consider $p'$ as part of $w$. {\bf goto} {\bf (ENTER)}.

\vspace{0.3cm}
\noindent {\bf Case \boldmath{$II$}: there is no outgoing edge of type \boldmath{$(0,0)$} at the end point of \boldmath{$p'$}}. By construction, any edge on $p'$, and in particular the last edge $e'$ on $p'$, is of type $(0,1)$. Note that at the moment $w$ violates the law of conservation only at the vertex $t(e')$, where there is one more unit of incoming flows than outgoing flows. Since $\Sigma$ is \AB, at the vertex $t(e')$, there must be one incoming edge $e_{11}$ of type $(1,1)$, one outgoing edge $e_{01}$ of type $(0,1)$, and another outgoing edge $e_{10}$ of type $(1,0)$, such that $e_{11}, e_{01}$ are consecutive edges on a path $p''$ of $w$ different from $p'$. See Figure~\ref{fig:trick_lemma} (Right). Now truncate $p''$ at $t(e')$ and append the second half of $p''$ to $p'$ so that $p'$ becomes a complete path ending at $B_2^c$ (still possibly with some type $(1,1)$ edges). Denote the first half of $p''$ by $r$, which contains at least one type $(1,1)$ edge such as $e_{11}$. Find the first type $(1,1)$ edge, say $e''$, along $r$. As before, truncate $r$ at $s(e'')$, throw away the second half of $r$ from $w$ and subtract $(0,1)$ from the type of each edge on the second half, and set $p'$ to be the first half of $r$. Set $q$ to be any path of $v$ connecting $B_1$ to $s(e'')$. Now the number of type $(1,1)$ edges is at most $n-2$ (which means once we have reduced $n$ to $n=1$, then Case $II$ cannot not occur). {\bf goto} {\bf (ENTER)}.

\vspace{0.3cm}
\noindent {\bf (EXIT):} $p'$ must end at $B_2^c$ (in fact, $B_1$), since otherwise the path $q$ combined with all the edges of $p'$ picked up in Case $I$ would form a path from $B_1$ to $B_2$, contradicting the assumption that $S(B_1) = 0$. Now $w$ becomes a valid flow and still consists of $S(B_2)$ edge-disjoint paths, and furthermore, the number of type $(1,1)$ edges is at most $n-1$. The induction follows.

Note that the above procedure will always end up in {\bf (EXIT)} after finitely many steps since the graph is finite and the number of type $(1,1)$ edges always decrease in Case $II$.
\end{proof}

Let $\Sigma$ be an \AB network with a partition $\partial \Sigma = A \sqcup B \sqcup C \sqcup D $. Recall the definition of $-I_3$ in Section~\ref{subsec:background_definitions},
\begin{align}
-I_3(A:B:C):= S(AB) + S(AC)+  S(BC) - S(A) - S(B) - S(C) - S(ABC)\ .
\end{align}
We write $-I_3(A:B:C)$ as $-I_3^{\Sigma}(A:B:C)$ when there is more than one network present.

\begin{theorem}\label{thm:I3}
Let $\Sigma$ be an \AB network with $\partial \Sigma = A \sqcup B \sqcup C \sqcup D $, then $-I_3(A:B:C) \geq 0$.
\end{theorem}
\begin{proof}
Note that for any flow $v$ on $\Sigma$, $S_{\Sigma}(AB;v) = S_{\Sigma}(A;v) + S_{\Sigma}(B;v)$. Combined with Lemma~\ref{lem:Res}, it follows that $-I_3^{\Sigma}(A:B:C) = -I_3^{\Res(v;\Sigma)}(A:B:C)$ for any flow $v$. Hence, it suffices to prove nonnegativity of $-I_3$ for any residual network.  By the nesting property~\cite{Freedman:2016zud, Headrick:2017ucz}, there exists a flow $v$ which is maximal simultaneously on $A$, $AB$, and $ABC$. Hence by Lemma~\ref{lem:Res}, the residual network of $v$ has max flow equal to zero on $A$, $AB$, and $ABC$. Without loss of generality, we may simply assume $S(A) = S(AB) = S(ABC) = 0$ for $\Sigma$, and then $-I_3^{\Sigma} = S(BC) + S(AC) - S(B) - S(C)$.

Let $B_1 = AB$, $B_2 = C$, $B_3 = D$ as in Lemma~\ref{lem:trick}, and let $v$ be a collection of $S(B)$ edge-disjoint paths from $B$ to $A$. This is possible since $S(AB)= 0$ and hence any flow starting from $B$ must end at $A$. Then $v$ is reachable from $B_1$. By Lemma~\ref{lem:trick}, there exists a flow $\tilde{v}$ extending $v$ such that $\tilde{v}-v$ consists of $S(C)$ edge-disjoint paths. Since $S(ABC) = 0$, these paths will end either at $A$ or at $B$, but never at $D$. See Figure~\ref{fig:directed_I3}. Let $\tilde{v}_B$ (resp. $\tilde{v}_A$) be the subflow of $\tilde{v}$ consisting of the paths which end at $B$ (resp. $A$), then $\tilde{v} = \tilde{v}_A +\tilde{v}_B$, and we have
\begin{align*}
S(B) + S(C) &= S_{\Sigma}(B;v) + S_{\Sigma}(C; \tilde{v}-v) \\
            &= S_{\Sigma}(B;v) - S_{\Sigma}(A; \tilde{v}-v) - S_{\Sigma}(B; \tilde{v}-v)\\
            &= S_{\Sigma}(BC;\tilde{v}_A) + S_{\Sigma}(AC;\tilde{v}_B)\\
            &\leq S(BC) + S(AC)\ . \qedhere
\end{align*}
\end{proof}

\begin{figure}
\centering
\includegraphics[scale=1]{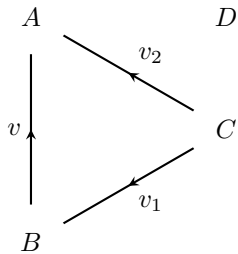}
\caption{A flow configuration resulting from the application of Lemma~\ref{lem:trick}. Here $\tilde{v} = v + v_1 + v_2$, $\tilde{v}_B = v_1$, $\tilde{v}_A = v + v_2$.}\label{fig:directed_I3}
\end{figure}

\begin{coro}\label{cor:general_I3}
Let $\Sigma = (V,E,c)$ be an \AB network with a rational capacity function such that $\partial \Sigma = A \sqcup B \sqcup C \sqcup D $, then $-I_3(A:B:C) \geq 0$. In particular, if $\Sigma$ is an \emph{undirected} network with a rational capacity function, then $-I_3(A:B:C) \geq 0$.
\end{coro}
\begin{proof}
The extension of nonnegativity of $-I_3$ from a constant capacity function to a rational capacity function is straightforward. One simply chooses an appropriate unit so that the rational function becomes integral, and then one splits every edge into several parallel edges, one for each unit capacity of that edge. The new edges all have capacity $1$. Apparently, the new network with the constant capacity function has the same max fluxes on any boundary as the original network.

If $\Sigma$ is an \emph{undirected} network, then by Section~\ref{subsec:background_definitions}, it can be viewed as a directed network with each edge replaced by a pair of parallel oppositely-oriented edges. Such a network is clearly \AB.
\end{proof}

Finally, we would like to point it out that the condition of being \AB is necessary for the nonnegativity of $-I_3$. Consider for instance the network as shown in Figure~\ref{fig:counter_example} which has three boundary vertices $A, B, C$, and we take $D$ to be empty. Then a straightforward computation shows that $-I_3(A:B:C) = -1$.
\begin{figure}
\centering
\includegraphics[scale=1]{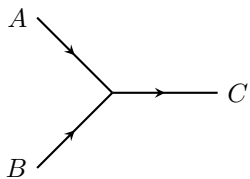}
\caption{A network with $-I_3 < 0$. $S(A) = S(B) = S(AB) = 1$ and all other maximal fluxes are zero.}\label{fig:counter_example}
\end{figure}

\section{Future directions}\label{sec:conclusion}

In this paper, our main goal was to prove the holographic entropy inequality MMI using bit threads. We successfully achieved this by proving Theorem~\ref{thm17}, which is esstentially the continuum generalization of Theorem~\ref{thm2} from graph theory. As far as we are aware, this is the first result concerning multicommodity flows in the setting of Riemannian manifolds. The proof itself may be of interest for many, as it borrows extensively tools from convex optimization. Indeed, using such tools, we were able to provide a novel proof for the old result Theorem~\ref{thm2} in graph theory as well. We hope that such tools can be fruitfully applied in the future to further our understanding of bit threads in holographic systems.

Our work leaves open several directions for further inquiry. The first is to understand the higher holographic entropy-cone inequalities found in~\cite{Bao:2015bfa} in terms of flows or bit threads. The full set of inequalities is known for five parties and conjectured for six; for more than six parties, only a subset of the inequalities is known. However, none of the inequalities beyond MMI follows from the known properties of flows and multiflows, such as nesting and Theorems~\ref{thm17} and~\ref{nmmf}. Therefore, flows must obey some \emph{additional} properties that guarantee those higher inequalities.%
\footnote{For example, one can prove the five-party cyclic inequality from~\cite{Bao:2015bfa} for networks by using similar techniques as presented in this paper by using a \emph{strengthened} version of Theorem~\ref{thm2}, known as the locking theorem (see e.g. \cite{frank1997integer}). Interestingly, the locking theorem does not appear to straightforwardly generalize to Riemannian geometries. (We thank V. Hubeny for pointing this out to us.)\label{locking footnote}}
Among the inequalities proved so far (subadditivity, strong subadditivity, MMI), each one has required a new property.
Clearly this is not very satisfying, and one can hope that there exists a unifying principle governing flows through which the full holographic entropy cone for any number of regions can be understood.

A second set of issues suggested by our work concerns the state decomposition conjecture of Section~\ref{sec:statedecomp}. There are really two problems here. Specifically, it would be useful both to sharpen the conjecture by constraining the possible form of the $1/N$ corrections and to find evidence for or against the conjecture. The simplest non-trivial case to test is the $n=3$ case with the regions $A$ and $C$ arranged so that $I(A:C)=0$ (at leading order in $1/N$). The ansatz~\eqref{threeregion} then simplifies to
\begin{equation}
\ket{\psi}_{ABC} = \ket{\psi_1}_{AB_1}\otimes\ket{\psi_3}_{B_3C} 
\end{equation}
(interpreted suitably, see Section~\ref{sec:statedecomp}).
Wormhole solutions of the kind studied in~\cite{Balasubramanian:2014hda} might be a particularly useful testing ground for these questions, since the corresponding state can be described in terms of a CFT path integral, potentially giving another handle on its entanglement structure.\footnote{We thank D.\ Marolf for useful discussions on this point.}

A third set of questions concerns a possible geometric decomposition of the bulk. These are motivated by our Theorem~\ref{thm4} and Conjecture~\ref{discconjn4} in the network case. Theorem~\ref{thm4} states that a network with three boundary vertices admits a decomposition into three subnetworks, effectively realizing the triangle skeleton diagram of Fig.~\ref{fig:skeleton}. Conjecture~\ref{discconjn4} makes a similar claim for a network with four boundary vertices. It would be very interesting from the point of view of graph theory to prove or disprove this conjecture. It would also be interesting to define an analogous decomposition in the Riemannian setting. Such a decomposition would imply that, for a given decomposition of the boundary, the bulk could be taken apart into building blocks. These would consist of a bridge connecting each pair of regions $A_i$, $A_j$ with capacity $\frac12I(A_i:A_j)$, and in the four-region case a four-way bridge realizing the star graph with capacity $-\frac12I_3$. If such a decomposition of the bulk can be defined and proved to exist, it would mirror the conjectured state decomposition. This would lead to the question of whether these two decompositions are physically related---is the bulk built up out of pieces representing elementary entanglement structures?

Fourth, it would be interesting to explore possible connections between our work and recent conjectures on entanglement of purification in holographic systems~\cite{Takayanagi:2017knl,Nguyen:2017yqw,Bao:2017nhh,Bao:2018gck,Kudler-Flam:2018qjo,Umemoto:2018jpc}. This actually involves two different issues. First, it seems reasonable to suppose that holographic entanglement of purification admits a description in terms of bit threads \cite{Du:2019emy}. Second, one could ask whether the entanglement of purification conjecture has any bearing on our state-decomposition conjecture or vice versa.

Finally, bit threads can be generalized to the covariant setting, where they reproduce the results of the HRT formula~\cite{covariant}. Since the MMI inequality is known to be obeyed by the HRT formula~\cite{Wall:2012uf}, it would be interesting to understand how bit threads enforce MMI in the covariant setting.

We leave all of these explorations to future work.

\section*{Acknowledgements}

We would like to thank David Avis, Ning Bao, Veronika Hubeny, and Don Marolf for useful conversations.
S.X.C. acknowledges the support from the Simons Foundation.
P.H. and M.H. were supported by the Simons Foundation through the ``It from Qubit'' Simons Collaboration as well as, respectively, the Investigator and Fellowship programs.
P.H. acknowledges additional support from CIFAR\@.
P.H. and M.W. acknowledge support by AFOSR through grant FA9550-16-1-0082.
T.H. was supported in part by DOE grant DE-FG02-91ER40654, and would like to thank Andy Strominger for his continued support and guidance.
T.H. and B.S. would like to thank the Okinawa Institute for Science and Technology for their hospitality, where part of this work was completed.
M.H. was also supported by the NSF under Career Award No.~PHY-1053842 and by the U.S. Department of Energy under grant
DE-SC0009987.
M.H. and B.S. would like to thank the MIT Center for Theoretical Physics for hospitality while this research was undertaken.
M.H. and M.W. would also like to thank the Kavli Institute for Theoretical Physics, where this research was undertaken during the program ``Quantum Physics of Information''; KITP is supported in part by the National Science Foundation under Grant No. NSF PHY-1748958.
The work of B.S. was supported in part by the Simons Foundation, and by the U.S. Department of Energy under grant DE-SC-0009987.
M.W. also acknowledges financial support by the NWO through Veni grant no.~680-47-459.
M.W. would also like to thank JILA for hospitality while this research was undertaken.

\end{spacing}

\bibliography{mmi-bib}

\begin{thebibliography}{10}
\providecommand{\url}[1]{{#1}}
\providecommand{\urlprefix}{URL }
\expandafter\ifx\csname urlstyle\endcsname\relax
  \providecommand{\doi}[1]{DOI~\discretionary{}{}{}#1}\else
  \providecommand{\doi}{DOI~\discretionary{}{}{}\begingroup
  \urlstyle{rm}\Url}\fi

\bibitem{Bakhmatov:2017ihw}
Bakhmatov, I., Deger, N.S., Gutowski, J., Colgain, E.O., Yavartanoo, H.:
  {Calibrated Entanglement Entropy}.
\newblock JHEP \textbf{07}, 117 (2017).
\newblock \doi{10.1007/JHEP07(2017)117}

\bibitem{Balasubramanian:2014hda}
Balasubramanian, V., Hayden, P., Maloney, A., Marolf, D., Ross, S.F.:
  {Multiboundary Wormholes and Holographic Entanglement}.
\newblock Class. Quant. Grav. \textbf{31}, 185,015 (2014).
\newblock \doi{10.1088/0264-9381/31/18/185015}

\bibitem{Bao:2017nhh}
Bao, N., Halpern, I.F.: {Holographic Inequalities and Entanglement of
  Purification}.
\newblock JHEP \textbf{03}, 006 (2018).
\newblock \doi{10.1007/JHEP03(2018)006}

\bibitem{Bao:2018gck}
Bao, N., Halpern, I.F.: {Conditional and Multipartite Entanglements of
  Purification and Holography}.
\newblock Phys. Rev. \textbf{D99}(4), 046,010 (2019).
\newblock \doi{10.1103/PhysRevD.99.046010}

\bibitem{Bao:2015bfa}
Bao, N., Nezami, S., Ooguri, H., Stoica, B., Sully, J., Walter, M.: {The
  Holographic Entropy Cone}.
\newblock JHEP \textbf{09}, 130 (2015).
\newblock \doi{10.1007/JHEP09(2015)130}

\bibitem{Boyd2004:CO:993483}
Boyd, S., Vandenberghe, L.: Convex Optimization.
\newblock Cambridge University Press, New York, NY, USA (2004)

\bibitem{Casini:2008wt}
Casini, H., Huerta, M.: {Remarks on the entanglement entropy for disconnected
  regions}.
\newblock JHEP \textbf{03}, 048 (2009).
\newblock \doi{10.1088/1126-6708/2009/03/048}

\bibitem{Chandra}
Chandrasekaran, R.: Multicommodity maximum flow problems.
\newblock \url{https://www.utdallas.edu/~chandra/documents/networks/net7.pdf}

\bibitem{cherkassky1977}
Cherkassky, B.V.: A solution of a problem on multicommodity flows in a network.
\newblock Ekonomika i matematicheski motody \textbf{13}, 143--151 (1977)

\bibitem{ding2016conditional}
Ding, D., Hayden, P., Walter, M.: {Conditional Mutual Information of Bipartite
  Unitaries and Scrambling}.
\newblock JHEP \textbf{12}, 145 (2016).
\newblock \doi{10.1007/JHEP12(2016)145}

\bibitem{Du:2019emy}
Du, D.H., Chen, C.B., Shu, F.W.: {Bit threads and holographic entanglement of
  purification}  (2019)

\bibitem{elias1956note}
Elias, P., Feinstein, A., Shannon, C.E.: A note on the maximum flow through a
  network.
\newblock Information Theory, IRE Transactions on \textbf{2}(4), 117--119
  (1956)

\bibitem{Federer74}
Federer, H.: Real flat chains, cochains and variational problems.
\newblock Indiana Univ. Math. J. \textbf{24}, 351--407 (1974/75)

\bibitem{ford1956maximal}
Ford, L.R., Fulkerson, D.R.: Maximal flow through a network.
\newblock Canadian journal of Mathematics \textbf{8}(3), 399--404 (1956)

\bibitem{frank1997integer}
Frank, A., Karzanov, A.V., Sebo, A.: On integer multiflow maximization.
\newblock SIAM Journal on Discrete Mathematics \textbf{10}(1), 158--170 (1997)

\bibitem{Freedman:2016zud}
Freedman, M., Headrick, M.: {Bit threads and holographic entanglement}.
\newblock Commun. Math. Phys. \textbf{352}(1), 407--438 (2017).
\newblock \doi{10.1007/s00220-016-2796-3}

\bibitem{HL}
Harvey, R., Lawson Jr., H.B.: Calibrated geometries.
\newblock Acta Math. \textbf{148}, 47--157 (1982).
\newblock \doi{10.1007/BF02392726}.
\newblock \urlprefix\url{http://dx.doi.org/10.1007/BF02392726}

\bibitem{Hayden:2011ag}
Hayden, P., Headrick, M., Maloney, A.: {Holographic Mutual Information is
  Monogamous}.
\newblock Phys. Rev. \textbf{D87}(4), 046,003 (2013).
\newblock \doi{10.1103/PhysRevD.87.046003}

\bibitem{hayden2016holographic}
Hayden, P., Nezami, S., Qi, X.L., Thomas, N., Walter, M., Yang, Z.:
  {Holographic duality from random tensor networks}.
\newblock JHEP \textbf{11}, 009 (2016).
\newblock \doi{10.1007/JHEP11(2016)009}

\bibitem{Headrick:2013zda}
Headrick, M.: {General properties of holographic entanglement entropy}.
\newblock JHEP \textbf{03}, 085 (2014).
\newblock \doi{10.1007/JHEP03(2014)085}

\bibitem{covariant}
Headrick, M., Hubeny, V.E.: Covariant bit threads.
\newblock to appear

\bibitem{Headrick:2017ucz}
Headrick, M., Hubeny, V.E.: {Riemannian and Lorentzian flow-cut theorems}.
\newblock Class. Quant. Grav. \textbf{35}(10), 105,012 (2018).
\newblock \doi{10.1088/1361-6382/aab83c}

\bibitem{Cuenca:2019uzx}
Hernández~Cuenca, S.: {The Holographic Entropy Cone for Five Regions}  (2019)

\bibitem{Hubeny:2018bri}
Hubeny, V.E.: {Bulk locality and cooperative flows}.
\newblock JHEP \textbf{12}, 068 (2018).
\newblock \doi{10.1007/JHEP12(2018)068}

\bibitem{Karzanov1978}
Karzanov, A., V.~Lomonosov, M.: Systems of flows in undirected networks pp.
  59--66 (1978).
\newblock In Book: Matematicheskoe Programmirovanie i dr. (Engl.: Mathematical
  Programming, and etc.), O.I. Larichev, ed., Inst. for System Studies (VNIISI)
  Press, Moscow, 1978, Issue 1, pp.59-66, in Russian.

\bibitem{Kudler-Flam:2018qjo}
Kudler-Flam, J., Ryu, S.: {Entanglement negativity and minimal entanglement
  wedge cross sections in holographic theories}  (2018)

\bibitem{Kupershtokh1971}
Kupershtokh, V.L.: A generalization of the ford-fulkerson theorem to multipole
  networks.
\newblock Cybernetics \textbf{7}(3), 494--502 (1971).
\newblock \doi{10.1007/BF01070459}.
\newblock \urlprefix\url{https://doi.org/10.1007/BF01070459}

\bibitem{MR0437391}
Lov\'asz, L.: On some connectivity properties of {E}ulerian graphs.
\newblock Acta Math. Acad. Sci. Hungar. \textbf{28}(1--2), 129--138 (1976).
\newblock \doi{10.1007/BF01902503}.
\newblock
  \urlprefix\url{https://doi-org.resources.library.brandeis.edu/10.1007/BF01902503}

\bibitem{Maldacena:2013xja}
Maldacena, J., Susskind, L.: {Cool horizons for entangled black holes}.
\newblock Fortsch. Phys. \textbf{61}, 781--811 (2013).
\newblock \doi{10.1002/prop.201300020}

\bibitem{Guyslain}
Naves, G.: Notes on the multicommodity flow problem.
\newblock \url{http://assert-false.net/callcc/Guyslain/Works/multiflows}

\bibitem{Nezami:2016zni}
Nezami, S., Walter, M.: {Multipartite Entanglement in Stabilizer Tensor
  Networks}  (2016)

\bibitem{Nguyen:2017yqw}
Nguyen, P., Devakul, T., Halbasch, M.G., Zaletel, M.P., Swingle, B.:
  {Entanglement of purification: from spin chains to holography}.
\newblock JHEP \textbf{01}, 098 (2018).
\newblock \doi{10.1007/JHEP01(2018)098}

\bibitem{MR1088184}
Nozawa, R.: Max-flow min-cut theorem in an anisotropic network.
\newblock Osaka J. Math. \textbf{27}(4), 805--842 (1990).
\newblock \urlprefix\url{http://projecteuclid.org/euclid.ojm/1200782678}

\bibitem{pastawski2015holographic}
Pastawski, F., Yoshida, B., Harlow, D., Preskill, J.: {Holographic quantum
  error-correcting codes: Toy models for the bulk/boundary correspondence}.
\newblock JHEP \textbf{06}, 149 (2015).
\newblock \doi{10.1007/JHEP06(2015)149}

\bibitem{Ryu:2006ef}
Ryu, S., Takayanagi, T.: {Aspects of Holographic Entanglement Entropy}.
\newblock JHEP \textbf{08}, 045 (2006).
\newblock \doi{10.1088/1126-6708/2006/08/045}

\bibitem{Ryu:2006bv}
Ryu, S., Takayanagi, T.: {Holographic derivation of entanglement entropy from
  AdS/CFT}.
\newblock Phys. Rev. Lett. \textbf{96}, 181,602 (2006).
\newblock \doi{10.1103/PhysRevLett.96.181602}

\bibitem{schrijver2003combinatorial}
Schrijver, A.: Combinatorial optimization: polyhedra and efficiency, vol.~24.
\newblock Springer Science \& Business Media (2003)

\bibitem{MR700642}
Strang, G.: Maximal flow through a domain.
\newblock Math. Programming \textbf{26}(2), 123--143 (1983).
\newblock \doi{10.1007/BF02592050}.
\newblock \urlprefix\url{http://dx.doi.org/10.1007/BF02592050}

\bibitem{MR2685608}
Sullivan, J.M.: A crystalline approximation theorem for hypersurfaces.
\newblock ProQuest LLC, Ann Arbor, MI (1990).
\newblock
  \urlprefix\url{http://gateway.proquest.com/openurl?url_ver=Z39.88-2004&rft_val_fmt=info:ofi/fmt:kev:mtx:dissertation&res_dat=xri:pqdiss&rft_dat=xri:pqdiss:9110403}.
\newblock Thesis (Ph.D.)--Princeton University

\bibitem{Takayanagi:2017knl}
Umemoto, K., Takayanagi, T.: {Entanglement of purification through holographic
  duality}.
\newblock Nature Phys. \textbf{14}(6), 573--577 (2018).
\newblock \doi{10.1038/s41567-018-0075-2}

\bibitem{Umemoto:2018jpc}
Umemoto, K., Zhou, Y.: {Entanglement of Purification for Multipartite States
  and its Holographic Dual}.
\newblock JHEP \textbf{10}, 152 (2018).
\newblock \doi{10.1007/JHEP10(2018)152}

\bibitem{VanRaamsdonk:2010pw}
Van~Raamsdonk, M.: {Building up spacetime with quantum entanglement}.
\newblock Gen. Rel. Grav. \textbf{42}, 2323--2329 (2010).
\newblock \doi{10.1007/s10714-010-1034-0, 10.1142/S0218271810018529}.
\newblock [Int. J. Mod. Phys.D19,2429(2010)]

\bibitem{Wall:2012uf}
Wall, A.C.: {Maximin Surfaces, and the Strong Subadditivity of the Covariant
  Holographic Entanglement Entropy}.
\newblock Class. Quant. Grav. \textbf{31}(22), 225,007 (2014).
\newblock \doi{10.1088/0264-9381/31/22/225007}

\end{thebibliography}
\bibliographystyle{spmpsci}

\end{document}